\documentclass[aps,superscriptaddress,twocolumn,showpacs,prapplied]{revtex4-2}

\usepackage{amsmath}
\usepackage{latexsym}
\usepackage{amssymb}
\usepackage{graphicx}
\usepackage[colorlinks=true, citecolor=blue, urlcolor=blue]{hyperref}
\usepackage{float}
\usepackage{amsfonts}
\usepackage{textcomp}
\usepackage{mathpazo}
\usepackage{blindtext}

\usepackage{bbm}

\usepackage{xcolor}
\definecolor{myurlcolor}{rgb}{0,0,1}
\definecolor{mycitecolor}{rgb}{0,0.5,0}
\definecolor{myrefcolor}{rgb}{0.5,0,0}
\usepackage{hyperref}
\hypersetup{colorlinks,
linkcolor=myrefcolor,
citecolor=mycitecolor,
urlcolor=myurlcolor}

\usepackage[draft]{fixme}
\usepackage{amsmath,bbm}
\usepackage{scrextend}
\usepackage{graphicx}
\usepackage{amsfonts}
\usepackage{amssymb}
\usepackage{amsmath,amssymb,amsthm,verbatim,graphicx,bbm}
\usepackage{mathrsfs}
\usepackage{mathtools}
\usepackage{color,xcolor,longtable}

\def\ra{\rangle}
\def\la{\langle}

\def\be{\begin{equation}}
\def\ee{\end{equation}}
\def\ben{\begin{eqnarray}}
\def\een{\end{eqnarray}}
\def\eea{\end{array}}
\def\bea{\begin{array}}

\newcommand{\Tr}[1]{\mathrm{Tr}#1}
\newcommand{\bei}{\begin{itemize}}
\newcommand{\eei}{\end{itemize}}
\newcommand{\ket}[1]{|#1\rangle}
\newcommand{\bra}[1]{\langle#1|}

\newcommand{\proj}[1]{\ket{#1}\!\bra{#1}}
\newcommand{\braket}[2]{\langle{#1}|{#2}\rangle}

\def\ra{\rangle}
\def\la{\langle}

\newcommand{\I}{\mathbbm{1}}

\newcommand{\dl}[1]{\left|\!\left|#1\right|\!\right|}

\renewcommand{\emph}[1]{\textbf{#1}}

\makeatletter
\newtheorem*{rep@theorem}{\rep@title}
\newcommand{\newreptheorem}[2]{%
\newenvironment{rep#1}[1]{%
 \def\rep@title{#2 \ref{##1}}%
 \begin{rep@theorem}}%
 {\end{rep@theorem}}}
\makeatother

\theoremstyle{plain}
\newtheorem{thm}{Theorem}
\newtheorem*{thm*}{Theorem}
\newreptheorem{thm}{Theorem}
\newtheorem{lem}{Lemma}
\newtheorem{fakt}{Fact}

\theoremstyle{definition}

\theoremstyle{remark}

\usepackage[T1]{fontenc}

\begin{document}


\title{Almost device-independent certification of multipartite quantum states with minimal measurements}
\author{Shubhayan Sarkar}
\email{shubhayan.sarkar@ug.edu.pl}
\affiliation{Institute of Informatics, Faculty of Mathematics, Physics and Informatics,
University of Gdansk, Wita Stwosza 57, 80-308 Gdansk, Poland}

\author{Alexandre C. Orthey, Jr.}
\affiliation{Center for Theoretical Physics, Polish Academy of Sciences, Aleja Lotnik\'{o}w 32/46, 02-668 Warsaw, Poland}
\affiliation{Institute of Fundamental Technological Research, Polish Academy of Sciences, Pawi\'nskiego 5B, 02-106 Warsaw, Poland}

\author{Gautam Sharma}
\affiliation{Center for Theoretical Physics, Polish Academy of Sciences, Aleja Lotnik\'{o}w 32/46, 02-668 Warsaw, Poland}

\author{Saronath Halder}
\affiliation{Center for Theoretical Physics, Polish Academy of Sciences, Aleja Lotnik\'{o}w 32/46, 02-668 Warsaw, Poland}

\author{Remigiusz Augusiak}
\affiliation{Center for Theoretical Physics, Polish Academy of Sciences, Aleja Lotnik\'{o}w 32/46, 02-668 Warsaw, Poland}
\email{augusiak@cft.edu.pl}

\begin{abstract}	

Device-independent certification of quantum states enables the characterization of states within a device under minimal physical assumptions. A major problem in this regard is to certify quantum states using minimal resources. Aiming to address this problem, we consider a multipartite quantum steering scenario involving an arbitrary number of parties, of which only one is trusted, meaning that the measurements performed by this party are known. Consequently, the self-testing scheme is almost device-independent. Importantly, all the parties can only perform two measurements each, which is the minimal number of measurements required to observe any form of quantum nonlocality. Then, we propose steering inequalities that are maximally violated by three major classes of genuinely multipartite entangled (GME) states: graph states of arbitrary local dimension, Schmidt states of arbitrary local dimension, and $N$-qubit generalized W states. Using the proposed inequalities, we then provide an almost device-independent certification of the above GME states. Restricting to qubits, we also lift our almost device-independent scheme to device-independent self-testing.

\end{abstract}


\maketitle

\section{Introduction}

Device-independent (DI) quantum information refers to a framework in quantum information theory where protocols and tasks are designed without relying on detailed knowledge or assumptions about the devices used to implement them. In other words, DI protocols aim to achieve certain quantum information processing tasks without specifying the internal workings or characteristics of the quantum devices involved. This approach is particularly valuable for enhancing security and reliability in quantum communication and computation tasks. By removing assumptions about the devices, DI protocols provide a higher level of security against potential attacks and imperfections in the devices themselves. One of the key resources in DI quantum information processing is Bell nonlocality \cite{Bell,Bell66}, which serves as a fundamental ingredient for DI tasks such as quantum cryptography \cite{DICrypto}, random number generation \cite{di4}, and state and measurement certification \cite{Yang, di3, di4}. 

The most complete form of DI certification is referred to as self-testing \cite{10.5555/2011827.2011830} (see also Ref. \cite{SupicReview}).It allows one to infer the structure of the underlying quantum system solely from the observed Bell-type correlations, without requiring any knowledge of the measurement implementation. Recently, a great deal of attention has been devoted to finding schemes to self-test entangled quantum states in the bipartite \cite{Scarani, Bamps, Projection, sarkar, Jed1} and multipartite \cite{Flavio, Mckague_2014, Wu_2014, remik23,sarkar2, Santos_2023, Allst, sarkar2023universal,NPJQI} frameworks. In particular, \cite{Allst} provides a scheme for self-testing any pure multipartite entangled quantum state and \cite{sarkar2023universal} provides a scheme to self-test any quantum state. Although both results are general, they are based on quantum networks which significantly impacts their 
practical applicability.

A major problem in this regard is to find the optimal self-testing scheme in terms of the number of measurements required to certify a particular quantum state from the observed correlations. Clearly, as self-testing is based on quantum nonlocality, the minimal number of measurements per party is two. In the bipartite regime, there are only a few schemes that utilize two measurements per observer \cite{Scarani, Bamps, sarkar} and even fewer schemes exist in the multipartite scenario \cite{Flavio, sarkar2}. Moreover, most of the strategies existing in the multipartite scenario are devoted to quantum states that are locally qubits. One of the main challenges lies in designing a suitable Bell inequality that is maximally violated by a given multipartite entangled quantum state. This task becomes even more complex when one aims to construct such an inequality using only two measurements per observer.

To reduce complexity, one can introduce physically motivated assumptions into the considered scenario, leading to the concept of semi-device-independent certification. For example, one may assume that one of the parties is trusted, meaning that their measurement devices are fully characterized. This leads to scenarios typically referred to as one-sided device-independent (1SDI), where the key resource enabling certification is a weaker form of quantum nonlocality known as quantum steering \cite{Wiseman, Ecaval, Quin}. Recent investigations have explored the potential of such scenarios for self-testing quantum states and measurements \cite{Supic, Alex, Bharti, Chen, sarkar6, sarkar3, Sarkar_2023,sarkar2023network}. For this work, we used the multipartite quantum steering framework introduced in \cite{Cavalcanti_2015}.  

Here, we address a significantly more challenging problem: certifying genuinely multipartite entangled (GME) states shared among an arbitrary number of parties with arbitrary local dimensions, using only a minimal number of measurements, namely, two per party. To simplify the problem we make the aforementioned assumption that a single party is trusted, thus, making our scheme almost device-independent. Furthermore, we restrict the parties to only choose two measurements each. Under these assumptions, we construct steering inequalities that are maximally violated by three relevant classes of GME states. We then prove that maximal violation of our inequalities can be used for certification 
of the corresponding states. Since, among the arbitrary number of parties only one is trusted, the presented results below are almost DI certificates of the quantum state. The first family of states we can certify in this way are graph states of arbitrary local dimension. Graph states form one of the most representative classes of states, which includes for instance the well-known Greenberger-Horne-Zeilinger (GHZ) \cite{GreenbergerHorneZeilinger1989GoingBeyond}, cluster \cite{PhysRevLett.86.910} or the absolutely maximally entangled (AME) \cite{helwig2013absolutelymaximallyentangledqudit} states. Graph states are key resources for many applications, just to mention measurement-based quantum computing \cite{PhysRevLett.86.5188,Raussendorf2003MBQC}, quantum metrology \cite{PhysRevLett.124.110502} or multipartite secret sharing \cite{PhysRevA.59.1829,PhysRevA.78.042309}.
The existing self-testing schemes for graph states are restricted to local dimension two \cite{Mckague_2014, Flavio} and three \cite{Santos_2023}. Here we go beyond these particular dimensions and provide strategies that allow one to certify graph states of arbitrary local dimension, yet at the cost that one of the measurement devices is trusted.

The second class of states we consider here are the Schmidt states of arbitrary local dimension, which also include the aforementioned GHZ states. The Schmidt states are useful for quantum metrology \cite{PhysRevA.82.012337} or multipartite secret sharing \cite{PhysRevA.59.1829}. A self-testing method for this class was provided in Ref. \cite{remik23}. The latter, however, requires more measurements per observer and thus is not optimal as far as physical implementations are concerned. Interestingly, our scheme can also be utilized to certify two measurements with the untrusted parties to be mutually unbiased using almost no entanglement across multiple parties. 

We finally provide a self-testing scheme for the generalized $W$ states (see below for a definition) shared among an arbitrary number of parties but only with local dimension two.
The generalized $W$ states belong together with the Schmidt states to a broader class of entangled symmetric multipartite states and find applications for instance in quantum metrology \cite{PhysRevA.78.042309}. While self-testing methods for the standard $N$-qubit $W$ state were already introduced in Refs. \cite{Wu_2014,remik23}, here, by making a mild assumption that only one party is trusted, we can significantly generalise these results to a multi-parameter class of genuinely entangled states. 

Finally, we show how the above results can be extended to the device-independent regime by taking inspiration from Ref. \cite{Allst}. More precisely, we consider an extra party which enables certification of the measurements of the trusted party from violation of the CHSH Bell inequality. 



\section{Preliminaries}

We consider a one-sided DI (1SDI) scenario consisting of $N$-parties, $A$ and $B_1,\ldots ,B_{N-1}$, which are also called Alice and Bobs. All the parties are located in space-like separated laboratories and share a quantum state $\rho_{A\mathbf{B}}$, where we denote $\mathbf{B}\equiv B_1\ldots B_{N-1}$. Each party has a measurement device performing measurements on their shares of $\rho_{A\mathbf{B}}$. We additionally assume that Alice's measurement device is trusted, meaning that it performs known measurements. Yet, we do not assume that her device performs full tomography over her share of $\rho_{A\mathbf{B}}$. On the other hand, the measurement devices that belong to $N-1$ Bobs are untrusted and we treat them as ``black boxes''. The scenario we just described is also called \textit{local steering} \cite{RevModPhys.92.015001}, in which the untrusted parties perform local measurements in their respective states to steer the state owned by the trusted party. 

In this work, we consider that the trusted party performs the measurements given by the $d-$dimensional generalization of the Pauli $Z$ and $X$ observables,
\begin{align}\label{Aliceobservables}
    Z_d=\sum_{i=0}^{d-1}\omega^i\ket{i}\!\bra{i},\qquad X_d=\sum_{i=0}^{d-1}\ket{i+1}\!\bra{i}
\end{align} 
which, for ease of notation, will be denoted as $Z$ and $X$, respectively. Fig. \ref{fig:scheme} illustrates the scenario considered. 

We denote the inputs and outputs of the parties $A$ and $B_1,\ldots,B_{N-1}$ with $x$ and $y_1,\ldots, y_{N-1}$ and $a$ and $b_1,\ldots,b_{N-1}$, respectively. For convenience, we use the notation $\mathbf{y}\equiv y_1,\ldots,y_{N-1}$ and $\mathbf{b}\equiv b_1,\ldots,b_{N-1}$. The correlations observed by the parties when they repeat their measurements many times are captured by the set of probability distributions 
\begin{equation}
\vec{p}=\{p(a,\mathbf{b}|x,\mathbf{{y}})\},
\end{equation}
where $p(a,\mathbf{b}|x,\mathbf{y})$ stands for the probability of obtaining results $a$ and $\mathbf{b}$ given that the measurements $x$ and $\mathbf{y}$ have been carried out. According to Born's rule, the latter is expressed as
\begin{align}\label{correlations}
    p(a,\mathbf{b}|x,\mathbf{y})=\Tr\left[\rho_{A\mathbf{B}}\left(M_{a|x}\otimes \bigotimes_{i=1}^{N-1} N^{(i)}_{b_i|y_i}\right)\right],
\end{align}
where $M_{a|x}$ and $N^{(i)}_{b_i|y_i}$ are the measurement operators representing the measurements of Alice and all Bobs $B_i$, respectively; recall that they are positive semi-definite and sum to identity on their corresponding Hilbert spaces.
\begin{figure}[t]
    \centering
    \includegraphics[width=\linewidth]{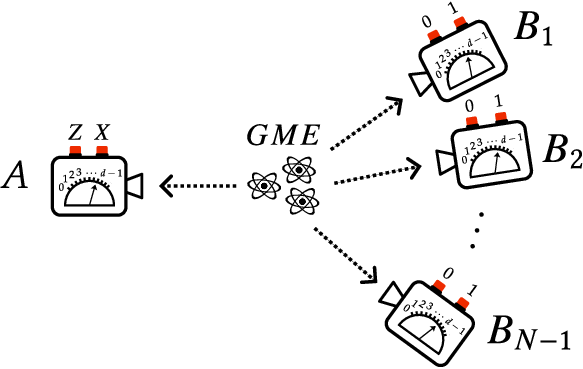}
    \caption{Depiction of the local steering scenario that we consider. A GME state is shared among $N$ parties: one Alice and $N-1$ Bobs. While Alice performs measurements $Z$ and $X$, all the other Bobs perform unknown measurements.}
    \label{fig:scheme}
\end{figure}

In our multipartite 1SDI scenario, the presence of local quantum steering is demonstrated if the probability distribution $p(a,\mathbf{b}|x,\mathbf{y})$ cannot be described by a local hidden state (LHS) model \cite{PhysRevLett.98.140402,RevModPhys.92.015001}, i.e., if it cannot be written as
\begin{align}\label{LHSmodel}
    p(a,\mathbf{b}|x,\mathbf{y})=\sum_{\lambda}p(\lambda)\Tr\left[M_{a|x}\rho_{A}^{\lambda}\right]p(\mathbf{b}|\mathbf{y},\lambda),
\end{align}
where $\lambda$ is a hidden variable that follows the probability distribution $p(\lambda)$. The hidden variable $\lambda$ determines the state $\rho^{\lambda}_A$ received by Alice and also the responses $p(\mathbf{b}|\mathbf{y},\lambda)$ obtained by $N-1$ Bobs. If the set of probability distributions $\{p(a,\mathbf{b}|x,\mathbf{y})\}$ cannot be described by the above model, it leads to violation of the linear steering inequality $\mathcal{B}\left[\{p(a,\mathbf{b}|x,\mathbf{y})\}\right]\leqslant \beta_L$, where $\mathcal{B}$ is a linear combination of $\{p(a,\mathbf{b}|x,\mathbf{y})\}$ and $\beta_L$ is the LHS bound, also called the classical bound \cite{PhysRevA.80.032112}, which is the maximal value of 
the functional $\mathcal{B}$ over the LHS models \eqref{LHSmodel}.

In this work, we express steering inequalities in terms of generalized expectation values, which are defined through $p(a,\mathbf{b}|x,\mathbf{y})$ in the following way
\begin{align}
\la A_{k|x}\mathbf{B}_{\mathbf{l}|\mathbf{y}}\ra=\sum_{a,\mathbf{b}=0}^{d-1}\omega^{a k+b_1l_1+\ldots+b_{N-1}l_{N-1}}p(a,\mathbf{b}|x,\mathbf{y}),
\end{align}
where $\mathbf{B}_{\mathbf{l}|\mathbf{y}}\equiv(B_1)_{l_1|y_1}\ldots(B_{N-1})_{l_{N-1}|y_{N-1}}$ with $k,l_1,\ldots, l_{N-1}\in\{0,\ldots,d-1\}$, and $\omega=\exp(2\pi \mathbbm{i}/d)$. If the probabilities $p(a,\mathbf{b}|x,\mathbf{y})$ express through Eq. (\ref{correlations}), we can write the expectation values as
\be
\la A_{k|x}\mathbf{B}_{\mathbf{l}|\mathbf{y}}\ra=\Tr\left[\rho_{A\mathbf{B}}A_x^k \otimes \bigotimes_{i=1}^{N-1} (B_i)_{l_i|y_i} \right],
\ee
where $A^k_x$ is the $k$-th power of the unitary observable $A_x$, 
which, to recall, we assume to be the generalized Pauli observables $Z$ and $X$. Then, $(B_i)_{l_i|y_i}$ are operators representing the measurements of $B_i$, which are defined through the Fourier transforms of the measurement operators $N^{(i)}_{b_i|y_i}$, i.e., 
\begin{equation}
(B_i)_{l_i|y_i}=\sum_{b_i}\omega^{b_il_i}N^{(i)}_{b_i|y_i}.   
\end{equation}
In fact, since we can write $N^{(i)}_{b_i|y_i}$ as an inverse Fourier transform of $(B_i)_{l_i|y_i}$, we can also refer to $\{(B_i)_{l_i|y_i}\}_{l_i=0}^{d-1}$ as the measurements of the $i$-th Bob. Moreover, if $N_{y_i}^{(i)}$ are projective, then $(B_i)_{l_i|y_i}$ are the $l_i$-th power of a unitary observable $B_{i,y_i}\coloneqq (B_i)_{1|y_i}$, that is, $(B_i)_{l_i|y_i}=\left(B_{i,y_i}\right)^{l_i}$ \cite{Jed1}.

In what follows, we will present new steering inequalities that are maximally violated by three classes of GME states: the $N$-qudit graph states, the $N$-qudit Schmidt states and the generalized $N$-qubit $W$ states. Recall that a GME state is one that cannot be written as a tensor product of two other pure states corresponding to an arbitrary bipartition of all the parties into two disjoint and non-empty subsets. More precisely, an $N$-partite state $\ket{\psi_{A\mathbf{B}}}$ is GME if $\ket{\psi_{A\mathbf{B}}}\neq \ket{\alpha}_Q\otimes\ket{\beta}_{\overline{Q}}$ for any two states $\ket{\alpha}_Q$, $\ket{\beta}_{\overline{Q}}$ and any bipartition of the parties $A$ and $B_1,\ldots,B_{N-1}$ into two non-empty and disjoint sets $Q$ and $\overline{Q}$.


Maximal violations of steering inequalities play a critical role in the self-testing results presented below. When the observed correlations $\vec{p}$ exhibit a maximal violation, they allow us to infer both the observer's state and the performed measurements up to local isometries. With that in mind, let us formally define multipartite 1SDI self-testing similar to the one introduced in \cite{sarkar6}. 
We assume here that the measurements of the untrusted parties are projective based on dilation arguments. At the same time, we do not assume the shared state $\rho_{A\mathbf{B}}$ to be purifiable by the unstrusted Bobs as due to the Schmidt decomposition between $A$ and $\mathbf{B}$ this would imply that the dimension of $\mathbf{B}$'s joint Hilbert space is the same as that of Alice, i.e., $d$. Yet, we can always consider a purification $\ket{\psi}_{A\mathbf{B}E}$ with an extra system $E$ such that $\rho_{A\mathbf{B}}=\Tr_E(\proj{\psi}_{A\mathbf{B}E})$, which is not possessed by any of the parties. 
Alice's measurements are projective and fixed as $A_x$, while the measurements of Bobs are arbitrary and represented by the operators $(B_{i})_{l_i|y_i}$. The goal is to certify that the state shared between Alice and Bob and the measurements performed by Bobs are equivalent, up to certain equivalences, to a reference state $\ket{\psi'}_{A\mathbf{B}'}\in (\mathbbm{C}^d)^{\otimes N}$ and to the reference measurements represented by $(B_{i})_{l_i|y_i}'$ defined on $\mathbbm{C}^d$, respectively, which all give rise to the same correlations $\vec{p}$. More formally, we say that $\ket{\psi'}_{A\mathbf{B}'}$ and $(B'_{i})_{l_i|y_i}$ are certified from the observed $\vec{p}$ if there exists a unitary $U_{i}:\mathcal{H}_{B_i}\to\mathcal{H}_{B_i}$ for each Bob $B_i$ such that
%
    \begin{equation}
       \left(\I_A\otimes \bigotimes_{i=1}^{N-1}U_i\right)\ket{\psi}_{A\mathbf{B}E}=\ket{\psi'}_{A\mathbf{B}'}\otimes\ket{\phi}_{\mathbf{B}''E},
\end{equation}
and
\begin{equation}
        U_i\, (B_{i})_{l_i|y_i}\, U_i^{\dagger} = (B'_{i})_{l_i|y_i}\otimes \I_{B_i''},\qquad \forall i
    \end{equation}
%
where Bob's Hilbert spaces decompose into tensor products
$\mathcal{H}_{B_i}=(\mathbbm{C}^{d})_{B_i'}\otimes\mathcal{H}_{B_i''}$, where $\mathcal{H}_{B_{i}''}$
is an auxiliary Hilbert space of unknown dimension and $\ket{\phi}_{\mathbf{B}''E}$ is a junk state from $\mathcal{H}_{\mathbf{B}''}\otimes\mathcal{H}_E$. Let us remark here that the measurements can only be certified on the local supports of the state, thus, without loss of generality, we assume that the local density matrices of the joint state $\ket{\psi}_{A\mathbf{B}E}$ are full-rank. 

\section{1SDI certification of multipartite quantum states}
\subsection{Graph states}

\begin{figure}[t]
    \centering
    \includegraphics[width=\linewidth]{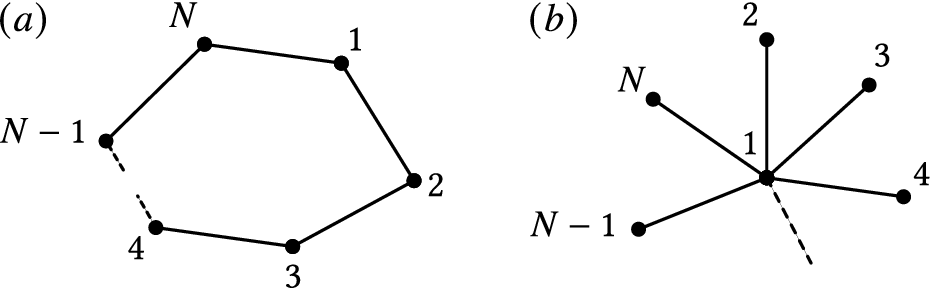}
    \caption{Examples of graph: (a) ring graph (b) star graph corresponding to $N-$qudit GHZ state up to local unitaries.}
    \label{fig:graph}
\end{figure}
Let us begin by defining graph states of arbitrary local dimension \cite{graphrev, graphrev2, graphrev3}. Consider a multigraph $G=(V,E,\Gamma)$ (cf. Fig. \ref{fig:graph}), where $V$ is a set of all the vertices of the graph, whereas $E$ is a set of pairs of vertices that are connected by edges. Finally, $\Gamma$ is a matrix with elements $\gamma_{i,j}\in\{0,\ldots,d-1\}$ specifying the multiplicities of the edges between vertices $i$ and $j$; notice that $\gamma_{i,j}=\gamma_{j,i}$ for any pair $i,j$ and $\gamma_{i,i}=0$. For our purpose, we consider graphs that are connected, that is, those in which there exists a path between any pair of vertices. Then, the graph state $\ket{G}\in(\mathbbm{C}^d)^{\otimes N}$ corresponding to the multigraph $G=(V,E,\Gamma)$, where $N$ is the number of vertices, is constructed as:
\begin{itemize}
    \item Each vertex of $G$ is associated with a single qudit state $\ket{+^d_i}=(1/\sqrt{d})\sum_{i=0}^{d-1}\ket{i}$, where $\{\ket{i}\}$ is the $d-$dimensional computational basis. The todal state of $N$ qudits is thus  
\begin{equation}
        \ket{+^d_V}=\bigotimes_{i=1}^N\ket{+_{i}^d}.
    \end{equation}
    \item For each edge $e=\{i,j\}$ connecting the vertices $i$ and $j$ of multiplicity $\gamma_{i,j}$, one applies the controlled unitary 
    \begin{equation}
        Z_{i,j}^{\gamma_{i,j}}=\sum_{q=0}^{d-1}\proj{q}_i\otimes \left(Z_j^{\gamma_{i,j}}\right)^q
    \end{equation}
is applied to the state $\ket{+^d_V}$.
    
\end{itemize}
Consequently, a graph state corresponding to the multigraph $G=(V,E)$ is defined as
\begin{equation}
    \ket{G}=\prod_{\{i,j\}\in E}Z_{i,j}^{\gamma_{i,j}}\ket{+^d_V}.
\end{equation}
An equivalent way to represent graph states is by using the stabilizer formalism \cite{graphuse4}. Let us consider a multigraph $G=(V,E)$ and associate to each 
vertex $i$ the following operator 
\begin{equation}\label{stab}
    \overline{S}(G)_{d,i}=X_i\otimes\bigotimes_{j\in n(i)} Z_j^{\gamma_{i,j}}\qquad  (i=1,2,\ldots,N),
\end{equation}
where $n(i)$ is the set of vertices connected to the vertex $i$, which is also called the neighborhood of $i$. The graph state $\ket{G}$ is the unique eigenstate of all the $N$ stabilizing operators $S_i(G)$ corresponding to the eigenvalue $1$. For example, the GHZ states are graph states.

Before presenting our first results, we note that fully device-independent schemes for graph states of local dimension two or three have already been proposed in Refs.~\cite{Mckague_2014, Flavio} and~\cite{Santos_2023}, respectively, whereas no scheme is known beyond these specific dimensions. Here, we provide strategies that allow for the certification of graph states of arbitrary local dimension, albeit under the assumption that one of the measurement devices is trusted.

Let us now move on to constructing steering inequalities that are maximally violated by graph states and then use these inequalities for self-testing. Let us consider a graph state $\ket{G}$ associated to a graph $G$. We construct the steering operator from the corresponding stabilizing operators \eqref{stab}. We then assume the first vertex to be trusted and label it as $0$; thus the measurement operators acting on the first site are known to be $Z_d$ and $X_d$ \eqref{Aliceobservables}. Next, we replace $Z_d$ and $X_d$ on the $j$-th vertex with arbitrary $d$-outcome observables $B_{j,0}$ and $B_{j,1}$, respectively. The steering operator is then given by
\begin{eqnarray}\label{graphsteeop_main}
    \mathcal{I}_d(G,N)&=&   X_A\otimes\bigotimes_{j\in n(0)}(B_{j,0})^{\gamma_{0,j}}\nonumber\\&&+\sum_{j\in n(0)}Z_{A}^{\gamma_{j,0}}\otimes B_{j,1}\otimes\bigotimes_{j'\in n(j)\setminus\{0\}}(B_{j',0})^{\gamma_{j,j'}}\nonumber\\
&&+\sum_{j\notin n(0)\cup {0}}B_{j,1}\otimes\bigotimes_{j'\in n(j)}(B_{j',0})^{\gamma_{j,j'}}+h.c.,\nonumber\\
\end{eqnarray}
where $h.c.$ denotes the Hermitian conjugate of the operator on its left-hand side. Notice that for a given graph $G=(V,E)$, the first term of the above operator corresponds to the stabilizer when choosing the vertex $0$, the second term consists of the stabilizers when choosing vertices connected to the $0$th vertex and the third term consists of rest of it. Also notice that in the above operator $\gamma_{i,j}=\gamma_{j,i}$ for any $i,j$. Thus, the steering inequality is given by
\begin{equation}\label{steingraph_main}
    \langle \mathcal{I}_d(G,N) \rangle\leqslant \beta_L,
\end{equation}
where $\beta_L$ is the maximum value attainable using the LHS model. In Fact 1 of Appendix \ref{AppendixA}, we provide an upper bound to $\beta_L$ in terms of a simple optimization problem that can be numerically evaluated for any graph $G$ and any number of parties $N$ and outcomes $d$. 

The quantum bound $\beta_Q$ of the steering functional $\langle \mathcal{I}_d(G,N) \rangle$ is equal to the number of terms it consists of, which is twice the number of the corresponding stabilizing operators, that is, $\beta_Q=2N$. Clearly, the quantum bound is attained by the graph state $\ket{G}$ and observables $B_{j,0}=Z$ and $B_{j,1}=X$ for any $j$. Importantly, it follows from Fact 1 of Appendix \ref{AppendixA} that $\beta_L<\beta_Q$, and consequently the proposed steering inequality is non-trivial for any graph state.

Now, if all the parties observe that $\langle \mathcal{I}_d(G,N) \rangle=2N$, then we can establish the following theorem. 
\begin{thm}
Assume that for a given graph $G$, the corresponding functional (\ref{graphsteeop_main})
reaches its maximal value $\langle\mathcal{I}_d(G,N)\rangle=2N$ for a state $\ket{\psi_{A\mathbf{B}E}}$ and observables $B_{i,y_i}$. Then, the following statement holds true for any $d$: $\mathcal{H}_{B_i}=(\mathbbm{C}^d)_{B'_i}\otimes\mathcal{H}_{B''_i}$, where $\mathcal{H}_{B''_i}$ is some finite-dimensional Hilbert space, and there exists a local unitary transformation on Bob's side $U_i:\mathcal{H}_{B_i}\rightarrow\mathcal{H}_{B_i}$, such that for any $i$,
%
%
\begin{equation}\label{Theo1.1m}
U_i\,B_{i,0}\,U_i^{\dagger}=Z\otimes\I_{B''_i},\quad U_i\,B_{i,1}\,U_i^{\dagger}=X\otimes\I_{B''_i},\ \ 
\end{equation}
where ${B''}$ denotes Bob's auxiliary system and the state $\ket{\psi}_{A\mathbf{B}E}$ is given by,
\begin{equation}\label{Theo1.2m}
\left(\I_{AE}\otimes \bigotimes_{i=1}^{N-1}U_i\right)\ket{\psi}_{A\mathbf{B}E}=\ket{G}_{A\mathbf{B}'}
\otimes\ket{\xi}_{\mathbf{B}''E},
\end{equation}
where $\ket{\xi}_{\mathbf{B}''E}$ denotes the auxiliary state corresponding to subsystems
$\mathbf{B}''$ and $E$.
\end{thm}

The proof of the above theorem is given in Appendix \ref{AppendixA}.

In practical scenarios, it is important to investigate the robustness of self-testing statements, such as those introduced in our work, against experimental imperfections, including noise in the preparation device and imperfect measurements. However, deriving analytical robustness bounds is typically a highly demanding task, especially for systems of arbitrary local dimension. In the simplest case of graph states of local dimension two, we can instead exploit Jordan’s lemma, which significantly simplifies the problem and allows us to obtain the following robustness bound.
%


\begin{thm} Assume that the steering inequality \eqref{steingraph} attains a value close to the maximal violation, $\langle\mathcal{I}_2(G,N)\rangle\geq \beta_Q-\varepsilon$, when the $0^\text{th}$ party Alice measures the observables $A_0=Z$ and $A_1=X$. Then, there exist local unitary operations $U_i$ such that  the following inequality is satisfied,
\begin{eqnarray}
    \dl{\ket{\psi'}_{A\mathbf{B}E}-\ket{G}_{A\mathbf{B}'}
\otimes\ket{\xi}_{\mathbf{B}''E}}\leq  (\sqrt{N}+g/2)\sqrt{\varepsilon},\nonumber\\
\end{eqnarray}
where, as before, 
\begin{equation}
 \ket{\psi'}_{A\mathbf{B}E}=\mathbbm{1}_{AE}\otimes \bigotimes_{i=1}^{N-1}U_i\ket{\psi} _{A\mathbf{B}E} , 
\end{equation}
$\ket{G}_{A\mathbf{B}'}$ is the graph state to be certified and $\ket{\xi}_{\mathbf{B}''E}$ is the auxiliary state. Finally $g$ is a factor depending on the particular graph (c.f. Theorem \ref{thm222} in Appendix \ref{AppendixA}).
\end{thm}
A detailed version of the theorem and the proof is given in Appendix \ref{AppendixA}.

\subsection{Schmidt states}
Every pure bipartite state $\ket{\psi}_{AB}$ admits a Schmidt decomposition, that is, there exist orthonormal bases $\ket{i}_A$ and $\ket{i}_B$ in $\mathcal{H}_A$ and $\mathcal{H}_B$, respectively, such that $\ket{\psi}_{AB}=\sum_i \alpha_i\ket{i}_A\otimes\ket{i}_B$, where $\alpha_i>0$ and $\sum_i \alpha_i^2=1$. 
An analogous decomposition in the multipartite setting is in general impossible. The multipartite pure states which can be represented via Schmidt decomposition are referred to as Schmidt states \cite{remik23}. Under the action of local unitary operators, an $N$-qudit Schmidt state $\ket{\psi(\pmb{\alpha})}_{A\mathbf{B}}$ can be brought to the following form 
\begin{align}\label{Schmidt}
\ket{\psi(\pmb{\alpha})}_{A\mathbf{B}}=\sum_{i=0}^{d-1}\alpha_{i}\ket{i}^{\otimes N},
\end{align}
$\pmb{\alpha}\equiv\alpha_1,\ldots,\alpha_{N}$ such that $0<\alpha_i<1$ for all $i$ and $\sum_{i=0}^{d-1}\alpha_i^2=1$. Because $\alpha_i>0$ for all $i$, the Schmidt rank of the state is maximal, which ensures that it is a truly $N$-qudit state. 

Importantly, a self-testing scheme for multipartite Schmidt states in the device-independent regime is already known and was proposed in Ref.~\cite{remik23}. That approach, however, requires more than two measurements per observer and is therefore suboptimal from an implementation perspective. Specifically, 
$N-1$ parties perform three measurements each, while the remaining party performs four. In this work, we exploit the assumption that one of the observers is trusted to reduce the number of measurements per observer to two.

Inspired by Ref. \cite{sarkar3}, let us now introduce a general class of steering inequalities that are maximally violated by the above Schmidt state for a given $\pmb{\alpha}$. For this purpose, we consider the following steering operator
\begin{eqnarray}\label{Stefn1_main}
\mathcal{S}_d(\pmb{\alpha},N)&=&\sum_{k=1}^{d-1}\left[\left(\sum_{i=1}^{N-1}A_0^{k}\otimes B_{i,0}^{k}\right)\right.\nonumber\\
&&+ \left.\gamma(\pmb{\alpha}) A_1^{k}\otimes\bigotimes_{i=1}^{N-1} B_{i,1}^{k}+\delta_k(\pmb{\alpha})A_0^{k}\right],
\end{eqnarray}
where the coefficients $\gamma(\pmb{\alpha})$ and $\delta_k(\pmb{\alpha})$ are given by
%
\begin{eqnarray}\label{gammadelta_main}
\gamma(\pmb{\alpha})&=&d\left[\sum_{\substack{i,j=0\\i\ne j}}^{d-1}\frac{\alpha_i}{\alpha_j}\right]^{-1},\\
\delta_k(\pmb{\alpha})&=&-\frac{\gamma(\pmb{\alpha})}{d}\sum_{\substack{i,j=0\\i\ne j}}^{d-1}\frac{\alpha_i}{\alpha_j}\omega^{k(d-j)}.
\end{eqnarray}
%
In our scenario, we consider that all Bobs perform two arbitrary measurements, denoted by $B_{i,0}$ and $B_{i,1}$, whereas Alice's measurement device is trusted and measures the generalized Pauli observables $A_0=Z$ and $A_1=X$.  

The corresponding steering inequality is given by 
\begin{eqnarray}\label{steinsch_main}
   \langle \mathcal{S}_d(\pmb{\alpha},N)\rangle\leqslant \beta_L.
\end{eqnarray}
Whereas it is a difficult task to analytically compute the maximal classical value of the expression $\beta_L$, in Fact 3 of Appendix \ref{AppendixB} we provide an upper bound on $\beta_L$ in terms of an optimization problem that can be numerically evaluated. Then, the quantum bound $\beta_Q$ of $\langle\mathcal{S}_d(\pmb{\alpha},N)\rangle$ is given by $\beta_Q=(d-1)(N-1)+1$ and can be attained by $\ket{\psi(\pmb{\alpha})}$ and observables $B_{j,0}=Z_d$ and $B_{j,1}=X_d$ for any $j$. Importantly, from Fact 3 we can also conclude that $\beta_L<\beta_Q$ and thus the proposed steering inequality is non-trivial.

Let us now state our result. 
\begin{thm}
Assume that the steering inequality \eqref{steinsch_main} is maximally violated, that is,
$\langle \mathcal{S}_d(\pmb{\alpha},N)\rangle=(d-1)(N-1)+1$. 
%
%
Then, the following statement holds true for any $d$:  $\mathcal{H}_{B_i}=(\mathbbm{C}^d)_{B'_i}\otimes\mathcal{H}_{B''_i}$, where $\mathcal{H}_{B''_i}$ is a finite-dimensional Hilbert space, and there exists a local unitary transformation on Bob's side $U_i:\mathcal{H}_{B_i}\rightarrow\mathbbm{C}^d\otimes\mathcal{H}_{B''_i}$ such that for any $i$,
\begin{equation}\label{Theo1.1schm}
U_i\,B_{i,0}\,U_i^{\dagger}=Z\otimes\I_{B''_i},\quad U_i\,B_{i,1}\,U_i^{\dagger}=X\otimes\I_{B''_i},
\end{equation}
where $B''_i$ denotes Bob's auxiliary system and the state $\ket{\psi}_{A\mathbf{B}}$ is given by
\begin{equation}\label{Theo1.2schm}
\left(\I_{AE}\otimes \bigotimes_{i=1}^{N-1}U_i\right)\ket{\psi_N}=\ket{\psi(\pmb{\alpha})}_{A\mathbf{B}'}
\otimes\ket{\xi}_{\mathbf{B}''E},
\end{equation}
where $\ket{\xi}_{\mathbf{B}''E}$ denotes the auxiliary state.
\end{thm}
The proof of the above theorem can be found in Appendix \ref{AppendixB}.

\subsection{$N$-qubit generalized $W$ states}

The third and last class of states that we consider comprises the $N$-qubit generalized $W$ states. As the name suggests, they represent a generalization of the well-known $W$ state $\ket{\psi_W}=(\ket{100}+\ket{010}+\ket{001})/\sqrt{3}$ and have the following form
\begin{align}\label{WN_main}
\ket{\psi_W(\pmb{\alpha})}=\sum_{l=0}^{N-1}\alpha_{l+1}\ket{0\ldots,0,1_l,0,\ldots,0},
\end{align}
such that $\alpha_i>0$ for every $i$ and $\sum_{i=1}^{N}\alpha_i^2=1$. We introduce a steering inequality that is maximally violated by the above state through the following steering operator
\begin{eqnarray}\label{steinWN_main}
\mathcal{W}_N(\pmb{\alpha})  &=& -2Z_A\otimes \bigotimes_{k=1}^{N-1}B_{k,0}+
\sum_{l=1}^{N-1}  \left[ Z_A \otimes \I_l\otimes\left(\I-P_l\right) \right.\nonumber\\
&&\left.   +\left(\gamma_l X_A \otimes B_{l,1}+\delta_lZ_A \otimes \I_{l}\right)\otimes P_l \right]\nonumber\\
&&+\sum_{l=1}^{N-1}  \left[ \I_A \otimes B_{l,0}\otimes\left(\I-P_l\right)    +\I_A \otimes \I_l\otimes P_l \right],\nonumber\\
\end{eqnarray}
where
\begin{align}\label{coeffN_main}
    \gamma_l=\frac{2\alpha_{l+1}\alpha_1}{\alpha_{l+1}^2+\alpha_1^2},\qquad \delta_l=\frac{\alpha_{l+1}^2-\alpha_1^2}{\alpha_{l+1}^2+\alpha_1^2},
\end{align}
for every $l\in \{1,\ldots,N-1\}$,
\begin{eqnarray}\label{Pl_main}
   P_l=\frac{1}{2^{N-2}} \bigotimes_{k=1, k \neq l}^{N-1}\left(\I_k+ B_{k,0}\right),
\end{eqnarray}
and $B_{k,0}$ and $B_{k,1}$ are the measurements acting on the $k$-th qubit. Similarly, $\I_l$ is the identity operator acting on the $l$-th qubit. The corresponding steering inequality is given by
\begin{eqnarray}\label{inequalityWN_main}
    \langle \mathcal{W}_N(\pmb{\alpha})\rangle\leqslant\beta_L.
\end{eqnarray}
While as before we are unable to compute the LHS bound $\beta_L$, we can prove the inequality to be nontrivial by demonstrating that observation of the maximal quantum value of $\langle \mathcal{W}_N(\pmb{\alpha})\rangle$, which is $\beta_Q=2N$ (see Appendix \ref{AppendixC} for a proof), allows one to self-test $\ket{\psi_W(\pmb{\alpha})}$. In fact, we can prove the following statement. 
%
%
\begin{thm}
Assume that $\langle \mathcal{W}_N(\pmb{\alpha})\rangle=2N$ for a state
$\ket{\psi}_{A\mathbf{B}E}$ and Bob's observables $B_{i,y_i}$. 
%
Then, $\mathcal{H}_{B_i}=\mathbbm{C}^2\otimes\mathcal{H}_{B''_i}$, where $\mathcal{H}_{B''_i}$ is some finite-dimensional Hilbert space, and there exists a local unitary transformation on Bob's side $U_i:\mathcal{H}_{B_i}\rightarrow\mathbbm{C}^2\otimes\mathcal{H}_{B''_i}$, such that
\begin{equation}\label{Theo3_measurements_main}
U_i\,B_{i,0}\,U_i^{\dagger}=Z\otimes\I_{B''_i},\quad U_i\,B_{i,1}\,U_i^{\dagger}=X\otimes\I_{B''_i}
\end{equation}
for every $i$, where ${B''}$ denotes Bob's auxiliary system and the state   $\ket{\psi}_{A\mathbf{B}E}$ is given by,
\begin{equation}\label{Theo3_states_main}
\left(\I_{AE}\otimes \bigotimes_{i=1}^{N-1}U_i\right)\ket{\psi}_{A\mathbf{B}E}=\ket{\psi_W(\pmb{\alpha})}_{A\mathbf{B}'}
\otimes\ket{\xi}_{\mathbf{B}''E},
\end{equation}
where $\ket{\xi}_{\mathbf{B}''E}$ denotes the auxiliary state.
\end{thm}
As before, the proof of this statement is deferred to Appendix C.

Let us notice that self-testing methods for the standard $N$-qubit $W$ state
(i.e., one for which $\alpha_l=1/\sqrt{N}$ for every $l=1,\ldots,N$) were introduced in Refs. \cite{Wu_2014,remik23}. Here, by making a mild assumption that a single party is trusted, we can significantly generalise these results to a multi-
parameter class of genuinely entangled states.

\section{Extending to device-independence}

We now show how our 1SDI certification methods for multipartite quantum states can be extended to the device-independent regime, restricting for simplicity to the case where the quantum states are locally qubits. To this end, we combine our quantum steering scenario with the network-assisted approach presented in Ref. \cite{Allst}. More precisely, we introduce an extra party Charlie in between Alice and Bobs and also an additional source of quantum particles distributing a state $\ket{\phi}_{AC}$ to Alice and Charlie (cf. Fig. \ref{fig:scheme}). The aim of Charlie is to play the CHSH game with Alice in order to certify her measurements and also to teleport the local state $\rho_A=\Tr_{\mathbf{B}}\rho_{A\mathbf{B}}$ of the multipartite state $\rho_{A\mathbf{B}}$ to Alice. Charlie performs three measurements (cf. Fig. \ref{fig2}), each having two outcomes. The correlations observed by all parties are described by 
$\{p(a,c,b_1,\ldots, b_{N-1}|x,z,y_1,\ldots
,y_{N-1})\}$ where $z\in\{0,1,2\}$ and $c\in\{0,1\}$ denote the input and output of Charlie, respectively. 
The first two Charlie's measurements labeled by $z=0,1$ are performed on Charlie's subsystem of $\ket{\phi_{AC}}$, whereas the third one on both Charlie's subsystems.
Let us now go through the steps in details.

\begin{figure}
    \centering
    \includegraphics[width=\linewidth]{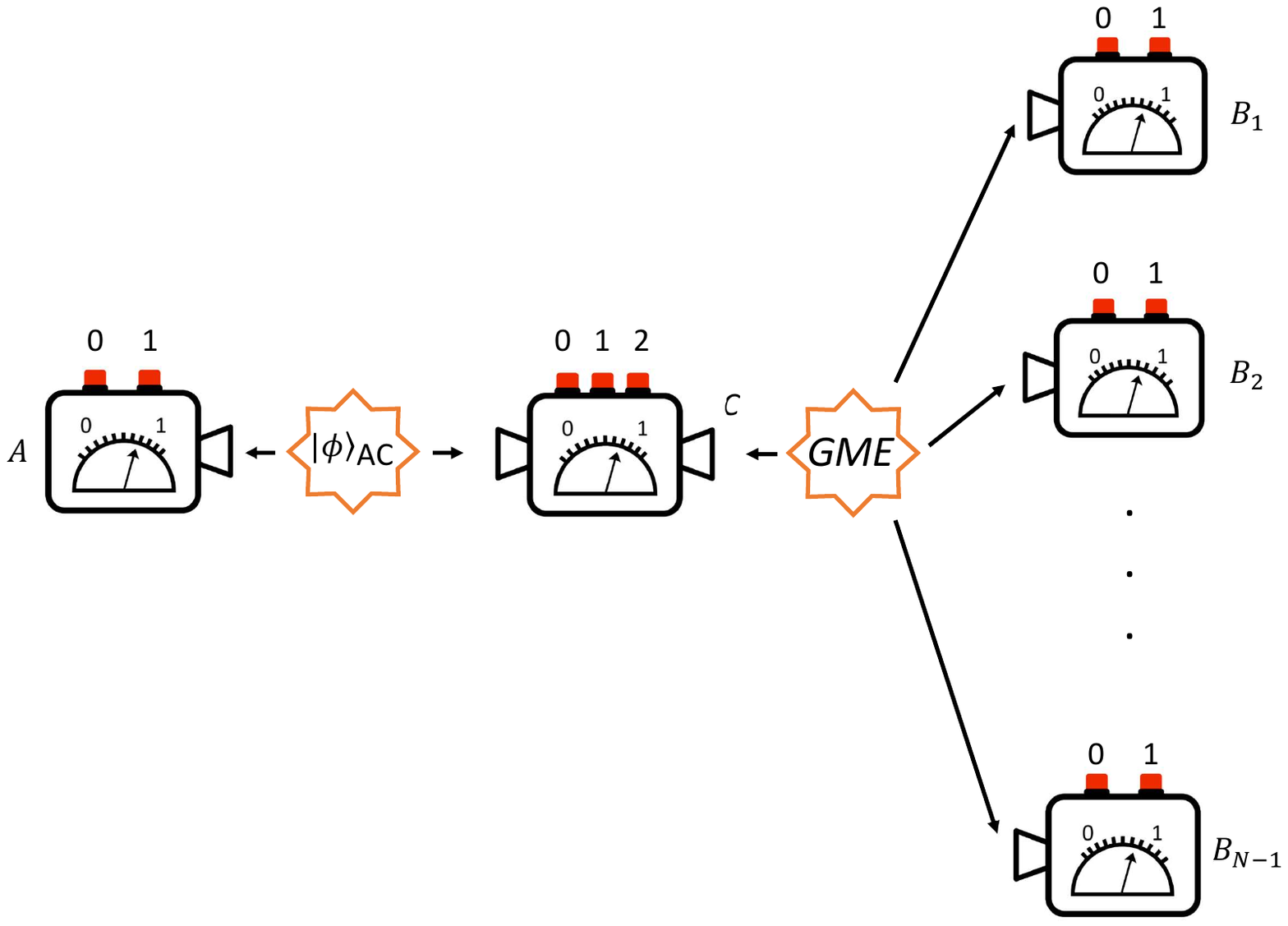}
    \caption{\textbf{Extended Bell scenario.} An additional party, Charlie, enforces trust in Alice by certifying his measurements. In the considered scenario there are two sources of states: one in between Alice and Charlie distributing bipartite states $\ket{\phi}_{AC}$ and one in between Charlie and $N-1$ Bob distributing $N$ partite states. }
    \label{fig2}
\end{figure}

\subsection*{Step I: Certifying Alice's observables}

To certify Alice's observables, Alice and Charlie simply play the Clauser-Horne-Shimony-Holt (CHSH) game \cite{CHSH} on the state distrubted by the bipartite source. If the correlations observed by them $\{p(a,c|x,y)\}$ violate the CHSH Bell inequality maximally, that is, they achieve the maximal quantum value of the functional 
\begin{equation}
    I_{\mathrm{CHSH}}:=\langle A_0C_0\rangle+\langle A_0C_1\rangle+\langle A_1C_0\rangle-\langle A_1C_1\rangle,
\end{equation}
which is $I_{\mathrm{CHSH}}=2\sqrt{2}$, then they infer that there exist local unitary operations $U_A$ and $U_C$ such that
\begin{eqnarray}\label{mea1}
    U_{A}\otimes U_{C}\ket{\psi}_{AC}=\ket{\phi^+}_{A'C'}\otimes\ket{\xi}_{A''C''},
\end{eqnarray}
where $\ket{\xi}_{A''C''}$ is some auxiliary state, and
\begin{eqnarray}\label{Alices}
U_{A}\,A_0\,U_A^{\dagger}=Z\otimes\I_{A''},\quad
    U_{B}\,A_1\,U_B^{\dagger}=X\otimes\I_{B''}.
\end{eqnarray}
Thus, Alice, by playing the CHSH game with Charlie, can certify her measurements to be, up to the above equivalences, the two Pauli observables which are used in the steering inequality \eqref{inequalityWN_main}.


\subsection*{Step II: Teleporting $\rho_{A}$ Charlie to Alice}

Once the state between Alice and Charlie is certified to be maximally entangled the local state $\rho_A$ of the multipartite state $\rho_{A\mathbf{B}}$ from Charlie to Alice using the third Charlie's measurement corresponding to $z=2$, which is performed on both systems received by Charlie. We only restrict our analysis to the first outcome of the Bell measurement, as no communication is allowed between Alice and Bob during the protocol. The reason for this teleportation step rather than directly sending the state $\rho_A$ to Alice is that the measurements of Alice are only certified on the local support of $\Tr_C\proj{\phi}_{AC}$. If the local support is different, then the measurement cannot be guaranteed [see \cite{sarkar3}]. Consequently, in the next step, we are only interested in the correlations $\{p(a,0,b_1,\ldots,b_N|x,2,y_1,\ldots,y_N)\}$, keeping in mind the self-tested state and measurements in the previous step.

\subsection*{Step-III: Self-testing with certified measurements}

In the last step Alice together with all Bobs can now use the 1SDI certification methods developed in the previous section to certify a given multipartite state in $\rho_{A\mathbf{B}}$. However, now the trusted observables $Z$ and $X$ must be replaced by those given in (\ref{Alices}), which have been previously certified by Alice and Charlie. 

Consider now an extended Bell scenario operator $\overline{\mathcal{S}}$ constructed, for instance, from the steering operator \eqref{steinWN_main} in which $Z$ and $X$ have been replaced by the measurements $A_1,A_2$ certified \eqref{mea1}. It is not difficult to observe that the new steering operator can be written as 
\begin{equation}
 \overline{\mathcal{S}}=U_{A}^{\dagger}(\mathcal{S}\otimes\I_{A''})U_{A},
\end{equation}
where $\mathcal{S}$ are any of the above-introduced steering operators. Consequently, the LHS and quantum bounds remain the same in the extended Bell and the standard steering scenario. Moreover, if $\langle\overline{\mathcal{S}}\rangle=\beta_Q$, then we readily have that $\langle\mathcal{S}\otimes\I_{A''}\rangle=\beta_Q$. Consequently, any state $\ket{\psi_{A\mathbf{B}}}$ attaining the maximal value of $\overline{\mathcal{S}}$ must also satisfy $\Tr(\mathcal{S}\rho_{A'\mathbf{B}})=\beta_Q$ where $\rho_{A'\mathbf{B}}=\Tr_{A''}(U_A\psi_{A\mathbf{B}})$. This is the exact condition that we utilise in the above 1SDI certification proofs, and thus it can be straightforwardly used here to obtain that 
\begin{eqnarray}
    U_{A}\otimes U_{\mathbf{B}}\ket{\psi_{A'\mathbf{B}E}}=\ket{\tilde{\psi}_{A'\mathbf{B'}}}\otimes\ket{\xi_{\mathbf{B''}E}}
\end{eqnarray}
where $\ket{\psi'}_{A\mathbf{B}}$ is the ideal state and the system $E$ is added to purify $\rho_{A'\mathbf{B}}$. Consequently, the state generated by the source is certified to be $U_{A}\otimes U_{\mathbf{B}}\ket{\psi_{A\mathbf{B}E}}=\ket{\tilde{\psi}_{A'\mathbf{B'}}}\otimes\ket{\xi_{A''\mathbf{B''}E}}$.
\section{Discussion}
In our work, we showed that while designing self-testing methods with a minimal number of measurements is highly challenging, introducing trust in one of the parties enables the construction of 1SDI schemes for a broad class of genuinely multipartite entangled (GME) states. We focused on graph states, Schmidt states, and generalised $W$ states. Another important finding is that using an arbitrarily low amount of entanglement it is possible to certify, in a single process, a pair of mutually unbiased operators for an arbitrary number of observers. 
In inspired by the results of Ref. \cite{Allst}, we also lifted the 1SDI certification to DI self-testing by including an additional party to certify the measurement of the trusted party. While in this work we have only considered the case when the trusted party performs two-outcome measurements, this approach technique can be straightforwardly generalised to a higher number of outcomes with the trusted party, provided a suitable Bell inequality exists which allows one to self-test the measurements of the trusted party. One can also utilise the network-based self-testing of any measurement presented in \cite{sarkar2023universal} to certify arbitrary Alice's measurements; however, such an approach will be highly complicated and difficult to implement.
We believe that our results open a path towards DI self-testing strategies based on the minimal number of measurements per party which are applicable to any genuinely entangled state. At the moment no such a scheme exists.  

Our work raises several follow-up questions. An obvious one is how to expand our approach to certify other classes of genuinely entangled multipartite states. In fact, it would be interesting to see whether the assumption considered here that one or few of the measurement devices are trusted allows one to design certification scheme for any GME state based on the minimal number of measurements per observes. A related question that arises here is whether the ideas presented above can be combined with the results of Ref. \cite{mancinska2021constantsized,volcic2023constantsized} to design self-testing schemes for multipartite states which are based on finite statistics. What is more, it would have been interesting to explore whether the 1SDI certification scheme for graph states of any local dimension provided here could be generalized to stabilizer subspaces (cf. Ref. \cite{PhysRevLett.125.260507}). At the same time, as far as the experimental testing of our schemes is concerned, it is crucial to understand how resilient they are to experimental errors. Additionally, we can use our self-testing results for secure multipartite randomness generation. 

\section{Acknowledgments}
 This project was funded within the QuantERA II Programme (VERIqTAS project) that has received funding from the European Union’s Horizon 2020 research and innovation programme under Grant Agreement No 101017733 and from the Polish National Science Center (project No 2021/03/Y/ST2/00175). Funded by the European Union. Views and opinions expressed are however those of the author(s) only and do not necessarily reflect those of the European Union or the European Commission. Neither the European Union nor the granting authority can be held responsible for them. SS also acknowledges the National Science Centre, Poland, grant Opus 25, 2023/49/B/ST2/02468.
RA also acknowledges received funding from the European Union’s Horizon Europe research and innovation programme under grant agreement No 101080086 NeQST.

\providecommand{\noopsort}[1]{}\providecommand{\singleletter}[1]{#1}%

\pagebreak

\onecolumngrid

\appendix


\setcounter{thm}{0}

\section{Self-testing any graph state}
\label{AppendixA}
Consider a connected graph $G$ with $N$ vertices and edges denoted as $e$ with their multiplicity denoted by $\gamma_{i,j}$. Then the graph state $\ket{G}$ corresponding to graph $G$ is the unique state that stabilizes the following $N$ operators
\begin{eqnarray}
     S(G)_{d,i}=X_i\otimes\bigotimes_{j\in n(i)} Z_j^{\gamma_{i,j}}\qquad  i=1,2,\ldots,N
\end{eqnarray}
 
 Let us now recall the steering operator introduced in the main text.
\begin{equation}\label{graphsteeop}
    \mathcal{I}_d(G,N)=X_A\otimes\bigotimes_{j\in n(0)}(B_{j,0})^{\gamma_{0,j}}+\sum_{j\in n(0)}Z_{A}^{\gamma_{j,0}}\otimes B_{j,1}\otimes\bigotimes_{j'\in n(j)\setminus\{0\}}(B_{j',0})^{\gamma_{j,j'}}+\sum_{j\notin n(0)\cup {0}}B_{j,1}\otimes\bigotimes_{j'\in n(j)}(B_{j',0})^{\gamma_{j,j'}}+h.c.
\end{equation}
where $h.c.$ denotes the Hermitian conjugate of the operator on its left-hand side. Also notice that in the above operator $\gamma_{i,j}=\gamma_{j,i}$ for any $i,j$ where $\gamma_{i,j}$ is the number of edges connecting vector $i,j$ and $n(j)$ is the set of vertices connected to $j$ vertex. The cardinality of $n(j)$ is denoted as $\overline{n}_j$. 
Now, the steering inequality is given by
\begin{equation}\label{steingraph}
    \langle \mathcal{I}_d(G,N) \rangle\leqslant \beta_L.
\end{equation}
The classical bound $\beta_L$ is given below.
\begin{fakt}
The local hidden state (LHS) bound $\beta_L$ of the steering functional $\langle \mathcal{I}_d(G,N)\rangle$ where $\mathcal{I}_d(G,N)$ is given in \eqref{graphsteeop} is upper bounded by  
\begin{eqnarray}
    \beta_L\leqslant\max_{\rho}\left(2|\langle X_A\rangle_{\rho}|+2\sum_{j\in n(0)}|\langle Z_{A}^{\gamma_{j,0}}\rangle_{\rho}|\right)+2(N-\overline{n}_0-1)
\end{eqnarray}
where $Z$ and $X$ are given in eq. \eqref{Aliceobservables}. Furthermore, $\beta_L<2N$.
\end{fakt}
\begin{proof}
As shown in \cite{sarkar6}, for any LHS model with trusted Alice, the expectation values of the joint observables take the form
\begin{eqnarray}\label{genexpLHS}
\langle A^k_{x}\otimes B^l_y\rangle_{LHS}=\sum_{\lambda}p(\lambda)\langle A^k_{x}\rangle_{\rho_{\lambda}}\langle B^l_y\rangle_\lambda.
\end{eqnarray}
Consequently, we can express the steering inequality $\langle \mathcal{I}_d(G,N)\rangle$, where $\mathcal{I}_d(G,N)$ is given in \eqref{graphsteeop}, for any LHS model as
\begin{eqnarray}
     \langle \mathcal{I}_d(G,N)\rangle_{LHS}=  \sum_{\lambda}p(\lambda)\langle \mathcal{I}_d(G,N)\rangle_{\lambda,LHS},
\end{eqnarray}
where $  \langle \mathcal{I}_d(G,N)\rangle_{\lambda,LHS}$ is
\begin{eqnarray}
     \langle \mathcal{I}_d(G,N)\rangle_{\lambda,LHS}=\langle X_A\rangle_{\rho_{\lambda}}\langle\bigotimes_{j\in n(0)}(B_{j,0})^{\gamma_{0,j}}\rangle_{{\lambda}}+\sum_{j\in n(0)}\langle Z_{A}^{\gamma_{j,0}}\rangle_{\rho_{\lambda}}\langle B_{j,1}\otimes\bigotimes_{j'\in n(j)\setminus\{0\}}(B_{j',0})^{\gamma_{j,j'}}\rangle_{{\lambda}}\nonumber\\+\sum_{j\notin n(0)\cup {0}}\langle B_{j,1}\rangle_{\rho_{\lambda}}\langle\bigotimes_{j'\in n(j)}(B_{j',0})^{\gamma_{j,j'}}\rangle_{{\lambda}}+\langle h.c\rangle_{LHS}.\qquad
\end{eqnarray}
As $\langle B^l_y\rangle_\lambda\leq1$, we obtain that
\begin{eqnarray}
     \langle \mathcal{I}_d(G,N)\rangle_{LHS}\leqslant \sum_{\lambda}p(\lambda)\left(2|\langle X_A\rangle_{\rho_{\lambda}}|+2\sum_{j\in n(0)}|\langle Z_{A}^{\gamma_{j,0}}\rangle_{\rho_{\lambda}}|\right)+2(N-\overline{n}_0-1),
\end{eqnarray}
which can be upper bounded as
\begin{eqnarray}
      \langle \mathcal{I}_d(G,N)\rangle_{LHS}\leqslant \sum_{\lambda}p(\lambda)\max_{\rho}\left(2|\langle X_A\rangle_{\rho}|+2\sum_{j\in n(0)}|\langle Z_{A}^{\gamma_{j,0}}\rangle_{\rho}|\right)+2(N-\overline{n}_0-1).
\end{eqnarray}
By using the fact that $\sum_{\lambda}p(\lambda)=1$, the above finally gives us
\begin{eqnarray}
     \langle \mathcal{I}_d(G,N)\rangle_{LHS}\leqslant \max_{\rho}\left(2|\langle X_A\rangle_{\rho}|+2\sum_{j\in n(0)}|\langle Z_{A}^{\gamma_{j,0}}\rangle_{\rho}|\right)+2(N-\overline{n}_0-1).
\end{eqnarray}
One can numerically evaluate the above quantity for any $d$ and any graph. Let us finally observe from the above expression that the LHS bound of the steering functional $\mathcal{I}_d(G,N)$ satisfies $\beta_L<2N$ because $Z$ and $X$ are mutually incompatible and there is no quantum state that can simultaneously stabilize $Z$ and $X$. 
\end{proof}
Let us now compute the quantum bound of the steering functional $\langle \mathcal{I}_d(G,N)\rangle$.
\begin{fakt}
    The quantum bound of the steering functional $\langle \mathcal{I}_d(G,N)\rangle$ where $\mathcal{I}_d(G,N)$ is given in \eqref{graphsteeop} is given by $\beta_Q=2N$ and is achieved by the graph state $\ket{G}$ with the observables 
    \begin{equation}
    B_{j,0}=Z,\qquad  B_{j,1}=X.
\end{equation}
\end{fakt}
\begin{proof}

Let us first observe that the number of terms in the steering operator \eqref{graphsteeop} is equal to twice the number of stabilizers $S(G)_{d,i}$ \eqref{stab}, that is, $2N$. Since the expectation value of each term is upper-bounded by $1$, thus, we have that $\langle \mathcal{I}_d(G,N)\rangle\leq2N$.
 
 Consider now that the unknown observables $B_{i,j}$ are given by
 \begin{eqnarray}
      B_{j,0}=Z,\qquad  B_{j,1}=X
 \end{eqnarray}
 along with the graph state $\ket{G}$, the state that is uniquely stabilized by the $S(G)_i$ operators. It is straightforward to observe that using these states and observables one can attain the quantum bound.    
\end{proof}

Now notice that the steering operator consists $2N$ terms and consequently, the algebraic bound of $\langle \mathcal{I}_d(G,N)\rangle$ is also $2N$ as all the operators $B^{\dagger}_{i,j}B_{i,j}=\I$. Thus, it is necessary that any state $\ket{\psi_N}$ and observables $B_{i,j}$ that achieves the quantum bound must satisfy the following conditions
\begin{equation}
    \langle\psi_N|\overline{S}(G)_i|\psi_N\rangle=1,\qquad i={1,2,\ldots,N},
\end{equation}
where $\overline{S}(G)_i$ represents each term of the steering operator \eqref{graphsteeop}. Now, using Cauchy-Schwarz inequality, one can conclude from here that
\begin{equation}\label{ststab}
  \forall i\qquad   \overline{S}(G)_i|\psi_N\rangle=|\psi_N\rangle.
\end{equation}
The above condition would be extremely useful for certifying any graph state. Let us now establish the following theorem.
\begin{thm}
Assume that the steering inequality \eqref{steingraph} is maximally violated when the $0^\text{th}$ party Alice chooses the observables $A_i$ defined to be $A_0=Z$, $A_1=X$ and all the other parties choose their observables as $B_{i,j}$ for  $i=1,2\ldots,N-1$ and $j\in \{0,1\}$ acting on $\mathcal{H}_{B_i}$. Let us say that the state which attains this violation is given by $\ket{\psi_N}_{A\mathbf{B}}\in\mathbbm{C}^d\otimes\bigotimes_{i=1}^{N-1}\mathcal{H}_{B_i}\otimes\mathcal{H}_E$, then the following statement holds true for any $d$:  $\mathcal{H}_{B_i}=\mathcal{H}_{B'_i}\otimes\mathcal{H}_{B''_i}$, where $\mathcal{H}_{B'_i} \equiv \mathbbm{C}^d$, $\mathcal{H}_{B''_i}$ is some finite-dimensional Hilbert space, and there exists a local unitary transformation on Bob's side $U_i:\mathcal{H}_{B_i}\rightarrow\mathcal{H}_{B_i}$, such that
\begin{eqnarray}\label{Theo1.1}
\forall j, \quad U_i\,B_{j,0}\,U_i^{\dagger}=Z_{B'_i}\otimes\I_{B''_i},\qquad U_i\,B_{j,1}\,U_i^{\dagger}=X_{B'_i}\otimes\I_{B''_i}
\end{eqnarray}
where ${B''}$ denotes Bob's auxiliary system and the state   $\ket{\psi_{N}}$ is given by,
\begin{eqnarray}\label{Theo1.2}
\left(\I_{AE}\otimes \bigotimes_{i=1}^{N-1}U_i\right)\ket{\psi_N}_{A\mathbf{B}E}=\ket{G}_{A\mathbf{B}'}
\otimes\ket{\xi}_{\mathbf{B}''E},
\end{eqnarray}
where $\ket{\xi}_{\mathbf{B}''E}\in\bigotimes_{i=1}^{N-1}\mathcal{H}_{B_i''}\otimes\mathcal{H}_E$ denotes the auxiliary state.
\end{thm}
\begin{proof}
Let us begin by considering the first term and second term of the steering operator $\eqref{graphsteeop}$ corresponding to the stabilizers when choosing the $0^\text{th}$ vertex and vertices connected to it respectively. As concluded in Eq. \eqref{ststab}, we have that
\begin{eqnarray}
X_A\otimes\bigotimes_{j\in n(0)}(B_{j,0})^{\gamma_{0,j}}\ket{\psi_N}=\ket{\psi_N},
\end{eqnarray}
and
\begin{eqnarray}
Z_{A}^{\gamma_{j,0}}\otimes B_{j,1}\otimes\bigotimes_{j'\in n(j)\setminus\{0\}}(B_{j',0})^{\gamma_{j,j'}}\ket{\psi_N}=\ket{\psi_N}\qquad \forall j\in n(0).
\end{eqnarray}
Let us now choose a particular $k\in n(0)$ and rewrite the above conditions as,
\begin{eqnarray}
X_A\otimes(B_{k,0})^{\gamma_{0,k}}\otimes\bigotimes_{j\in n(0)\setminus\{k\}}(B_{j,0})^{\gamma_{0,j}}\ket{\psi_N}=\ket{\psi_N},
\end{eqnarray}
and
\begin{eqnarray}
Z_{A}^{\gamma_{k,0}}\otimes B_{k,1}\otimes\bigotimes_{j'\in n(k)\setminus\{0\}}(B_{j',0})^{\gamma_{k,j'}}\ket{\psi_N}=\ket{\psi_N}.            
\end{eqnarray}
Since the observables $X$ and $Z$ are unitary and $\gamma_{0,k}=\gamma_{k,0}$, we have that
\begin{eqnarray}\label{stgraeq1}
(B_{k,0})^{\gamma_{0,k}}\otimes\bigotimes_{j\in n(0)\setminus\{k\}}(B_{j,0})^{\gamma_{0,j}}\ket{\psi_N}=X_A^{-1}\ket{\psi_N}
\end{eqnarray}
 and
 \begin{eqnarray}\label{stgraeq2}
B_{k,1}\otimes\bigotimes_{j'\in n(k)\setminus\{0\}}(B_{j',0})^{\gamma_{k,j'}}\ket{\psi_N}= Z_{A}^{-\gamma_{0,k}}\ket{\psi_N}.            
\end{eqnarray}
Now, multiplying $Z_{A}^{-\gamma_{0,k}}$ to the condition \eqref{stgraeq1} and then using Eq. \eqref{stgraeq2}, we get that
\begin{eqnarray}\label{stgraph3}
(B_{k,0})^{\gamma_{0,k}}B_{k,1}\otimes\left(\bigotimes_{j\in n(0)\setminus\{k\}}(B_{j,0})^{\gamma_{0,j}}\right)\left(\bigotimes_{j'\in n(k)\setminus\{0\}}(B_{j',0})^{\gamma_{k,j'}}\right)\ket{\psi_N}=Z_{A}^{-\gamma_{0,k}}X_A^{-1}\ket{\psi_N}.
\end{eqnarray}
Then, multiplying $X_A^{-1}$ to the condition \eqref{stgraeq2} and then using Eq. \eqref{stgraeq1}, we get that
\begin{eqnarray}\label{stgraph4}
B_{k,1}(B_{k,0})^{\gamma_{0,k}}\otimes\left(\bigotimes_{j'\in n(k)\setminus\{0\}}(B_{j',0})^{\gamma_{k,j'}}\right)\left(\bigotimes_{j\in n(0)\setminus\{k\}}(B_{j,0})^{\gamma_{0,j}}\right)\ket{\psi_N}= X_A^{-1}Z_{A}^{-\gamma_{0,k}}\ket{\psi_N}. 
\end{eqnarray}
Notice, that terms inside the large brackets in both the above formulas commute. Now, by using the fact that $\omega^{\gamma_{0,k}} Z_{A}^{-\gamma_{0,k}}X_A^{-1}= X_A^{-1}Z_{A}^{-\gamma_{0,k}}$, we can multiply \eqref{stgraph3} by $\omega^{\gamma_{0,k}}$ and subtract \eqref{stgraph4} to obtain that
\begin{eqnarray}
\left(\omega^{\gamma_{0,k}}(B_{k,0})^{\gamma_{0,k}}B_{k,1}-B_{k,1}(B_{k,0})^{\gamma_{0,k}}\right)\otimes\left(\bigotimes_{j\in n(0)\setminus\{k\}}(B_{j,0})^{\gamma_{0,j}}\right)\left(\bigotimes_{j'\in n(k)\setminus\{0\}}(B_{j',0})^{\gamma_{k,j'}}\right)\ket{\psi_N}=0.
\end{eqnarray}
As the terms inside the large brackets in the above expression are unitary and thus invertible, we arrive at a simple condition
\begin{eqnarray}
\left(\omega^{\gamma_{0,k}}(B_{k,0})^{\gamma_{0,k}}B_{k,1}-B_{k,1}(B_{k,0})^{\gamma_{0,k}}\right)\otimes\I_{A\mathbf{B}E\setminus\{B_k\}}\ket{\psi_N}=0.
\end{eqnarray}
By taking a partial trace over all the subsystems except the $k^\text{th}$ one, we get that
\begin{eqnarray}
\left(\omega^{\gamma_{0,k}}(B_{k,0})^{\gamma_{0,k}}B_{k,1}-B_{k,1}(B_{k,0})^{\gamma_{0,k}}\right)\rho_{k,N}=0,
\end{eqnarray}
where $\rho_{k,N}=\Tr_{A\mathbf{B}E\setminus\{B_k\}}(\proj{\psi_N})$. As discussed before, the local states of each party are full-rank and thus invertible. Consequently, we finally have that
\begin{eqnarray}\label{graphfinal}
\omega^{\gamma_{0,k}}(B_{k,0})^{\gamma_{0,k}}B_{k,1}=B_{k,1}(B_{k,0})^{\gamma_{0,k}}.
\end{eqnarray}
As proven in \cite{Jed1}, if two unitary observables $B_{k,0}$ and $B_{k,1}$, for which $B_{k,0}^d=B_{k,1}^d=\I$, satisfy relation \eqref{graphfinal}, then there exists a unitary $U_{k}:\mathcal{H}_{B_k}\rightarrow \mathcal{H}_{B_k}$ such that
\begin{eqnarray}
\quad U_k\,B_{k,0}\,U_k^{\dagger}=Z_{B'_k}\otimes\I_{B''_k},\qquad U_k\,B_{k,1}\,U_k^{\dagger}=X_{B'_k}\otimes\I_{B''_k}.
\end{eqnarray}
Similarly, we can reach the above conclusion for any vertex $j\in n(0)$. Using these derived observables, we then find the observables for vertices connected with any vertex $j\in n(0)$ in exactly the same manner as done above. Since any vertex is connected to at least one other vertex, we continue this until all the observables are found and we get that
\begin{eqnarray}\label{graphfinal1}
U_i\,B_{i,0}\,U_i^{\dagger}=Z_{B'_i}\otimes\I_{B''_i},\qquad U_i\,B_{i,1}\,U_i^{\dagger}=X_{B'_i}\otimes\I_{B''_i}\qquad\forall i. 
\end{eqnarray}
Let us now characterize the state $\ket{\psi_N}$ that achieves the quantum bound of the steering functional $\langle \mathcal{I}_d(G,N)\rangle$ \eqref{graphsteeop}. Notice now that by putting into the condition \eqref{ststab} the derived observables from \eqref{graphfinal1}, we get that
\begin{eqnarray}
\forall i\qquad   \overline{S}(G)_i|\psi_N\rangle=\left(\bigotimes_{i=1}^{N-1}U_i^{\dagger}\right) S(G)_{d,i,AB_1''\ldots B_N''}\otimes\I_{B_1''\ldots B_N''E}\left(\bigotimes_{i=1}^{N-1}U_i\right)|\psi_N\rangle=|\psi_N\rangle,
\end{eqnarray}
where $S(G)_{d,i}$ are the stabilizing operators of a graph state $\ket{G}$ of the form \eqref{stab}. As discussed above, the graph state $\ket{G}$ is the unique eigenstate of such operators with eigenvalue $1$. Consequently, we have
\begin{eqnarray}
\left(\bigotimes_{i=1}^{N-1}U_i\right)|\psi_N\rangle_{A\mathbf{B}E}=\ket{G}_{A\mathbf{B}'}
\otimes\ket{\xi}_{\mathbf{B}''E},
\end{eqnarray}
where $\ket{\xi}_{\mathbf{B}''E}\in\bigotimes_{i=1}^{N-1}\mathcal{H}_{B_i''}\otimes\mathcal{H}_E$ denotes the auxiliary state.
\end{proof}

Let us now restrict ourselves to $d=2$ and find the robust self-testing statement. Before proceeding, let us state a general Lemma which will be required for the robust proof.

\begin{lem}\label{lemma1}
 Consider two hermitian and unitary matrices $A_0,A_1$ acting on one the subsystems of a bipartite state $\ket{\psi}\in\mathcal{H}_A\otimes\mathcal{H}_B$. Assume then that $\dl{\{A_0,A_1\}\ket{\psi}}\leq\varepsilon_N$ for some $0<\varepsilon<1$, where we additionally assume that under Jordan's lemma both $A_0$ and $A_1$ consist of only nontrivial $2\times 2$ blocks. Then, there exists a local unitary transformation $U:
 \mathcal{H}_A\to \mathbbm{C}^2\otimes\mathcal{H}_{A''}$ such that  
 \begin{eqnarray}\label{cor11}
U\,A\,U^{\dagger}\ket{\psi'}=Z\otimes\I\ket{\psi'},\qquad  \dl{(U\, B\, U^{\dagger}-X\otimes\I)\ket{\psi'}}\leq 2\varepsilon,
 \end{eqnarray}
%
%
 where $\ket{\psi'}= U\ket{\psi}$, and for simplicity, we have omitted all the identities acting on the second subsystem $B$.
 \end{lem}
The proof of this lemma can be found in Appendix B of Ref. \cite{sarkar2025gap}. Let us now derive the robustness statement.

\begin{thm}\label{thm222} Let us consider a graph state $\ket{G}$ that we want to certify robustly. For this graph, we consider a hierarchy of nodes as: (0) $n(0)$: all nodes connected to the $0-$th node (trusted node), (1) $n_1(j)$: all nodes connected to some $j-$th such that $j\in n(0)$, which are not neighbors of the node $0$, nor the node $0$ itself, (2) $n_2(j)$: all nodes connected to some node in $n_1(j)$ such that they are not connected to $0$ and $n(j)$, etc. Since we consider graphs which are not disjoint, for each vertex there exists a path connecting it to the trusted node.

Assume now that the steering inequality \eqref{steingraph} attains a value close to the maximal violation $\langle\mathcal{I}_2(G,N)\rangle\geq \beta_Q-\varepsilon$, when the $0^\text{th}$ party (trusted) Alice chooses the observables $A_i$ defined to be $A_0=Z$, $A_1=X$. Then, the state $\ket{\psi}_{A\mathbf{B}}$ and the observables $B_{i,j}$ for $i=1,\ldots,N-1$ and $j=0,1$ are certified as
\begin{eqnarray}\label{CarelessWhisper}
    \dl{\ket{\psi'}_{A\mathbf{B}E}-\ket{G}_{A\mathbf{B}'}
\otimes\ket{\xi}_{\mathbf{B}''E}}\leq  \frac{2\sqrt{N}+g}{2}\sqrt{\varepsilon}
\end{eqnarray}
where $g=\sum_kn_ka_k$ is the number of nodes connected in the $k$th heirarchy and $a_k=(8/3)(4^{k+1}-1)$ 
and for node $i$ belonging to $k$th heirarchy,
\begin{eqnarray}\label{A412}
U_{B_i}\, B_{i,0}\,U_{{B_i}}^{\dagger}\ket{\psi'}=Z\otimes\I\ket{\psi'},\qquad  \dl{(U_{B_i}\, B_{i,1}\, U_{{B_i}}^{\dagger}-X\otimes\I)\ket{\psi'}}\leq a_k\sqrt{\varepsilon}, 
 \end{eqnarray}
 where $\ket{\psi'}=\bigotimes_{i=1}^{N-1}U_{{B_i}}\ket{\psi}$ and we have dropped the indices $A\mathbf{B}E$. 
    
\end{thm}
\begin{proof}
    Let us consider the steering operator $\mathcal{I}_2(G,N)$ given in Eq. \eqref{graphsteeop} which in the particular case of $d=2$ simplifies to 
\begin{eqnarray}\label{graphsteeop_maind2}
    \mathcal{I}_2(G,N)&=&  2\left[ X_A\otimes\bigotimes_{j\in n(0)}B_{j,0}+\sum_{j\in n(0)}Z_{A}\otimes B_{j,1}\otimes\bigotimes_{j'\in n(j)\setminus\{0\}} B_{j',0}+\sum_{j\notin n(0)\cup {0}}B_{j,1}\otimes\bigotimes_{j'\in n(j)}B_{j',0}\right],
\end{eqnarray}    
where the factor $2$ is a consequence of the fact that all local observables are hermitian in this case.

    Let us then observe that it admits the following sum-of-squares decomposition, $\beta_Q\I-\mathcal{I}_2(G,N)=\sum_{j=1}^N P_jP_j^{\dagger}$ where the operators $P_j$ are given by
\begin{eqnarray}\label{SOS_graph}
        P_0&=&\I-X_A\otimes\bigotimes_{i\in n(0)}B_{i,0},\nonumber\\
        P_j&=&\I-Z_{A}\otimes B_{j,1}\otimes\bigotimes_{j'\in n(j)\setminus\{0\}}B_{j',0}\quad \mathrm{for}\quad j\in n(0), \nonumber\\ P_j&=&\I-B_{j,1}\otimes\bigotimes_{j'\in n(j)}B_{j',0}\quad \mathrm{for}\quad j\notin n(0)\cup {0}.
    \end{eqnarray}
Consequently, when $\langle\mathcal{I}_2(G,N)\rangle\geq \beta_Q-\varepsilon$, we have that $\dl{P_j\ket{\psi}}\leq \sqrt{\varepsilon}$ for all $j$. 

Let us now choose the $k-th$ vertex such that $k\in n(0)$, take two of the above relations and expand using \eqref{SOS_graph} as
\begin{eqnarray}
    \dl{\ket{\psi}-X_A\otimes B_{k,0}\otimes\bigotimes_{j\in n(0)\setminus \{k\}}B_{j,0}\ket{\psi}}\leq\sqrt{\varepsilon},\quad  \dl{\ket{\psi}-Z_{A}\otimes B_{k,1}\otimes\bigotimes_{j'\in n(k)\setminus\{0\}}B_{j',0}\ket{\psi}}\leq\sqrt{\varepsilon}.\quad
\end{eqnarray}
Now, using the fact that $X,Z$ are hermitian and unitary for $d=2$ and that the vector norm is unitarily invariant, the above inequalities imply that
\begin{eqnarray}\label{A37}
     \dl{X_A\ket{\psi}-B_{k,0}\otimes\bigotimes_{j\in n(0)\setminus \{k\}}B_{j,0}\ket{\psi}}\leq\sqrt{\varepsilon},\qquad  \dl{Z_{A}\ket{\psi}-B_{k,1}\otimes\bigotimes_{j'\in n(k)\setminus\{0\}}B_{j',0}\ket{\psi}}\leq\sqrt{\varepsilon}.
\end{eqnarray}
Exploiting the unitary invariance of the vector norm once more we can multiply the expressions under the norm by $Z_{A}$ and $X_A$, respectively, to obtain
\begin{eqnarray}\label{rownanie1}
\dl{Z_{A}X_A\ket{\psi}-Z_{A}\otimes B_{k,0}\otimes\bigotimes_{j\in n(0)\setminus \{k\}}B_{j,0}\ket{\psi}}\leq\sqrt{\varepsilon},\nonumber\\  \dl{X_AZ_{A}\ket{\psi}-X_A\otimes B_{k,1}\otimes\bigotimes_{j'\in n(k)\setminus\{0\}}B_{j',0}\ket{\psi}}\leq\sqrt{\varepsilon}.
\end{eqnarray}

Let us now consider the second inequality in (\ref{A37})
and by using the unitary invariance of the norm as well as the fact that all observables are unitary, let us expand it to the following form
\begin{equation}\label{rownanie2}
\dl{Z_{A}\otimes B_{k,0}\otimes\bigotimes_{j\in n(0)\setminus \{k\}} B_{j,0}\ket{\psi}-B_{k,0} B_{k,1}\otimes\left(\bigotimes_{j'\in n(k)\setminus\{0\}}B_{j',0}\right)\left(\bigotimes_{j'\in n(k)\setminus\{0\}}B_{j',0}\right)\ket{\psi}}\leq\sqrt{\varepsilon}
\end{equation}
Then, we the aid of the triangle inequality for the vector norm, one shows that the first inequality in \eqref{rownanie1} and \eqref{rownanie2} imply that 
\begin{equation}
    \dl{Z_{A}X_A\ket{\psi}-B_{k,0}B_{k,1}\otimes\left(\bigotimes_{j\in n(0)\setminus \{k\}}B_{j,0}\right)\left(\bigotimes_{j'\in n(k)\setminus\{0\}}B_{j',0}\right)\ket{\psi}}\leq2\sqrt{\varepsilon}.
\end{equation}
Analogously, from by suitably modifying the first inequality 
in \eqref{A37} and combining it with the the second inequality in \eqref{rownanie1} through the triangle inequality, one can obtain
%
%
\begin{eqnarray}
%
\dl{X_AZ_{A}\ket{\psi}-B_{k,1}B_{k,0}\otimes\left(\bigotimes_{j\in n(0)\setminus \{k\}}B_{j,0}\right)\left(\bigotimes_{j'\in n(k)\setminus\{0\}}B_{j',0}\right)\ket{\psi}}\leq2\sqrt{\varepsilon}.
\end{eqnarray}
In deriving these inequalities we also exploited the fact that 
$B_{j,0}$ commute among each other if they act on different sites.

Now, using the fact that $Z_{A}X_A=-X_AZ_{A}$, and then exploiting the triangle inequality, we have that
\begin{eqnarray}\label{A43}
\dl{[B_{k,0}B_{k,1}+B_{k,1}B_{k,0}]\ket{\psi}}\leq4\sqrt{\varepsilon}\qquad (k=1,\ldots,n(0)).
\end{eqnarray}
To all the observables $B_{k,i}$ for $k\in n(0)$ and $i=0,1$ (which are those that are measured by the neighbors of the trusted party) we can apply Lemma \ref{lemma1} (where for simplicity we additionally assume that $B_{k,i}$ decompose under Jordan's lemma only into nontrivial $2\times 2$ blocks), to conclude that there exists a local unitary operation $U_i$ acting on $i$th node such that
%
%
 %
 %
 \begin{eqnarray}\label{A42}
U_{i}\, B_{i,0}\, U_{{i}}^{\dagger}\ket{\psi'}=Z\otimes\I\ket{\psi'},\qquad  \dl{(U_{i}\, B_{i,1}\, U_{{i}}^{\dagger}-X\otimes\I)\ket{\psi'}}\leq 8\sqrt{\varepsilon},
 \end{eqnarray}
where $\ket{\psi'}=\bigotimes_{i\in n(0)}U_{{i}}\ket{\psi}$. 

The above steps allow us to find the robust bounds to observables for the nodes $j$ directly connected to the trusted party $j\in n(0)$. Based on the relations in Eq. \eqref{A42}, we can follow the same reasoning as above to prove analogous relations for the observables measured by the neighbors of all nodes in $n(0)$ that do not themselves belong to $n(0)$ and are not the trusted party. To this aim, let us now consider one of such nodes $j\in n(0)$ and prove that relations similar to \eqref{A42} hold true for all nodes $k\notin n(0)\cup\{0\}$ connected to $j$. We can follow the same steps as above using the SOS $P_j$ and $P_k$ \eqref{SOS_graph} to get
\begin{eqnarray}
    \dl{\ket{\psi}-B_{k,1}\otimes B_{j,0}\otimes\bigotimes_{j'\in n(k)} B_{j',0}\ket{\psi}}\leq\sqrt{\varepsilon},\quad  \dl{\ket{\psi}-Z_{A}\otimes B_{j,1}\otimes B_{k,0}\bigotimes_{j'\in n(j)\setminus\{0,k\}}B_{j',0}\ket{\psi}}\leq\sqrt{\varepsilon}.\quad
\end{eqnarray}
As the observables $B_{j,0},B_{j,1}$ for $j\in n(0)$ are certified as in Eq. \eqref{A42}, this allows us to obtain 
\begin{eqnarray}
    \dl{\ket{\psi'}-B_{k,1}\otimes (Z\otimes\I)_j\otimes\bigotimes_{j'\in n(k)} B_{j',0}\ket{\psi'}}\leq\sqrt{\varepsilon},\quad  \dl{\ket{\psi'}-Z_{A}\otimes (X\otimes\I)_j\otimes B_{k,0}\bigotimes_{j'\in n(j)\setminus\{0,k\}}B_{j',0}\ket{\psi'}}\leq9\sqrt{\varepsilon}.\nonumber\\
\end{eqnarray}
Following the same steps from Eqs. \eqref{rownanie1} to \eqref{A43}, we obtain
\begin{eqnarray}\label{anticom2}
\dl{[B_{k,0}B_{k,1}+B_{k,1}B_{k,0}]\ket{\psi}}\leq20\sqrt{\varepsilon}\qquad (k\in n(j),\, k\notin n(0)\cup \{0\}).
\end{eqnarray}
This allows us to conclude from Lemma \ref{lemma1} that for each $i\in n(j), i\notin n(0)\cup \{0\}$, there exists a unitary operation $U_i$ such that
\begin{eqnarray}\label{A4211}
U_{i}\, B_{i,0}\, U_{{i}}^{\dagger}\ket{\psi''}=Z\otimes\I\ket{\psi''},\qquad  \dl{(U_{i}\, B_{i,1}\, U_{{i}}^{\dagger}-X\otimes\I)\ket{\psi''}}\leq 40\sqrt{\varepsilon},
 \end{eqnarray}
 where $\ket{\psi''}=\bigotimes_{i\in n(j)},i\notin n(0)\cup\{0\}U_{{i}}\ket{\psi}$. 

 We now define an hierarchy of nodes as: (0) $n(0)$: all nodes connected to the $0-$th node (trusted node), (1) $n_1(j)$: all nodes connected to some $j-$th such that $j\in n(0)$, which are not neighbors of the node $0$, nor the node $0$ itself, (2) $n_2(j)$: all nodes connected to some node in $n_1(j)$ such that they are not connected to $0$ and $n(j)$, etc. We proceed in this way until all vertices are taken care of; clearly, since we consider graphs which are not disjoint, for each vertex there exists a path connecting it to the trusted node. At each level the error in the relations \eqref{A42} and \eqref{anticom2} for the observables $B_{i,1}$ adds up as $8,40,168,\ldots$; one can check that this series is described $a_{k}=(8/3)(4^{k+1}-1)$ for $k-$th level of the hierarchy. More precisely, Bob's observables satisfy
\begin{equation}
    \dl{[B_{i,0}B_{i,1}+B_{i,1}B_{i,0}]\ket{\psi}}\leq a_k\sqrt{\varepsilon}\qquad (k\mathrm{th}\quad \mathrm{level}).
\end{equation}

 Let us again consider the steering operator $\mathcal{I}_2(G,N)$ \eqref{graphsteeop} and observe that $2\beta_Q[\beta_Q\I-\mathcal{I}_2(G,N)]=(\beta_Q\I-\mathcal{I}_2(G,N))^2$ where we used the fact that $\mathcal{I}_2(G,N)^2=\beta_Q\I$ where $\beta_Q=2N$. Consequently, from the assumption that $\langle\mathcal{I}_2(G,N)\rangle\geq \beta_Q-\varepsilon$, we infer the following inequality,
 \begin{eqnarray}
     \dl{[\beta_Q\I-\mathcal{I}_2(G,N)]\ket{\psi}}\leq 2\sqrt{N\varepsilon}.
 \end{eqnarray}
As $\mathcal{I}_2(G,N)$ is composed of the observables $B_{i,0},B_{i,1}$, which are certified as \eqref{A42}, we can rewrite the above formula using the triangle inequality as
\begin{eqnarray}\label{55}
    \dl{[\beta_Q\I-\mathcal{I}_2^{\mathrm{id}}\otimes\I_{B_1''\ldots B_{N-1}''}(G,N)]\ket{\psi'}}\leq 2\sqrt{N\varepsilon}+g\sqrt{\varepsilon}=\mu\sqrt{\varepsilon}
\end{eqnarray}
where $\mathcal{I}_2^{\mathrm{id}}=2\sum_{i=1}^{N}\overline{S}(G)_{2,i}$ [c.f. Eq. \eqref{stab}] and $\mu=2\sqrt{N}+g$ such that the factor $g$ takes into account the graph structure such that the errors add up based on the hierarchy of nodes defined below Eq. \eqref{A4211}.  Suppose that at $k$th heirarchy, there are $n_k$ nodes, then $g=\sum_kn_ka_k$.

Consider now the eigensystem of the ideal steering operator,
%
%
%
$\mathcal{I}_2^{\mathrm{id}}\ket{G_{l}}=\lambda_{l}\ket{G_{l}}$. It is simple to observe that the difference between the highest and second-highest eigenvalues is $4$ as the second-highest eigenvalue is attained by $\ket{G_{l}}$ where $l_k=1$ for some $k$ and $l_i=0$ for all $i$ such that $i\ne k$. Let us now express the above formula \eqref{55} using the eigensystem of $\mathcal{I}_l^{\mathrm{id}}$ as
\begin{eqnarray}\label{Aleja}
     \dl{\left[\sum_{s=1}^{2^N}(\beta_Q-\lambda_l)\proj{G_l}\otimes\I_{A_1''\ldots A_N''}\right]\ket{\psi'}}\leq \mu\sqrt{\varepsilon}. 
\end{eqnarray}

Let us now implement the following expansion of the 
state $\ket{\psi'}$
\begin{equation}
    \ket{\psi'}=\sum_{i=1}^{2^N}\alpha_i\ket{G_i}\ket{\xi_i},
\end{equation}
where $\ket{G_i}$ are the eigenstates of the ideal steering operator $\mathcal{I}^{\mathrm{id}}$,i.e., $\mathcal{I}^{\mathrm{id}}\ket{G_i}=\lambda_i\ket{G_i}$. Then, $\alpha_i$ are nonnegative coefficient such that $\sum_i\alpha_i^2=1$,
and $\ket{\xi_i}$ are some auxiliary vectors $\ket{\xi_i}\in \bigotimes_{i=1}^N\mathcal{H}_{B''_i}$. 
The fact that $\alpha_i$ can be taken nonnegative stems from 
the freedom in choosing the auxiliary vectors $\ket{\xi_i}$.
Let us also enumerate the eigenvalues $\lambda_i$ so that $\lambda_1=2N$ which is the quantum bound $\beta_Q$; then, the corresponding eigenstate $\ket{G_1}$ is the graph state to be self-tested, i.e., $\ket{G_1}\equiv \ket{G}$.

Using this decomposition, the expression
(\ref{Aleja}) allows one to deduce that
\begin{eqnarray}
 \dl{\sum_{l=1}^{2^N}(\beta_Q-\lambda_l)\alpha_l\ket{G_l}\ket{\xi_l}}\leq \mu\sqrt{\varepsilon}. 
\end{eqnarray}
which on expanding the left-hand side results in
\begin{eqnarray}
  \sum_{l=1}^{2^N}\alpha_l^2(\beta_Q-\lambda_l)^2\leq \mu^2\varepsilon.
\end{eqnarray}

Let us now consider the eigenvalues $\lambda_i$ of the ideal steering operator $\mathcal{I}^{\mathrm{id}}$. Clearly, the largest eigenvalue is exactly the quantum bound $\beta_Q=2N$
and therefore the corresponding term in the above sum vanishes.
It is then not difficult to observe that the next largest eigenvalue is $2N-4$, which allows us to write the following inequality 
\begin{equation}
    \sum_{l=1}^{2^N}\alpha_l^2(\beta_Q-\lambda_l)^2\geq 4\sum_{l=2}^{2N}\alpha_l^2,
\end{equation}
where we used the fact that $(\beta_Q-\lambda_l)^2\geq 4$ for every $l=2,\ldots,2^N$. Consequently, we have the following inequality
\begin{eqnarray}\label{robu32}
    \sum_{\substack{l=2}}^{2^N}\alpha_l^2\leq \mu^2\frac{\varepsilon}{4}.
\end{eqnarray}
We then use the fact that $\sum_{i}\alpha_i^2=1$ which together with (\ref{robu32}) implies that
$\alpha_{1}\geq\alpha_{1}^2\geq1-(18N)^2\frac{\varepsilon}{4}$.
Thus,
\begin{eqnarray}\label{65}
    \langle{\psi'}|(\ket{G}\ket{\xi_1})=\alpha_1\geq 1-\mu^2\frac{\varepsilon}{4},
\end{eqnarray}
which directly implies that
\begin{eqnarray}
    \dl{\ket{{\psi'}}-\ket{G}\ket{\xi_1}}=\sqrt{2\{1-\mathrm{Re}[ \langle \psi'|(\ket{G}\ket{\xi_1})]\}}=\sqrt{2(1-\alpha_1)}\leq 
    \frac{\mu}{2}\sqrt{\varepsilon},
\end{eqnarray}
which is the desired inequality \eqref{CarelessWhisper}. This completes the proof.
    
\end{proof}

\section{Self-testing Schmidt states}
\label{AppendixB}
The Schmidt states we consider have the form
\begin{align}\label{Schmidt_SM}
\ket{\psi(\pmb{\alpha})}_{A\mathbf{B}}=\sum_{i=0}^{d-1}\alpha_{i}\ket{i}^{\otimes N},
\end{align}
where $d=\min \{d_A,d_{B_i}\}$ and $\pmb{\alpha}$ is a vector composed of the Schmidt coefficients $\alpha_i$, such that $0<\alpha_i<1$ for all $i$ and $\sum_{i=0}^{d-1}\alpha_i^2=1$. Next, we are going to introduce a steering operator that provides an inequality maximally violated by states of the form \eqref{Schmidt_SM}.

Let us now recall the steering operator for Schmidt states introduced in the main text:
\begin{eqnarray}\label{Stefn1}
\mathcal{S}_d(\pmb{\alpha},N)=\sum_{k=1}^{d-1}\left[\left(\sum_{j=1}^{N-1}A_0^{k}\otimes B_{j,0}^{k}\right) + \gamma(\pmb{\alpha}) A_1^{k}\otimes\bigotimes_{j=1}^{N-1} B_{j,1}^{k}+\delta_k(\pmb{\alpha})A_0^{k}\right],
\end{eqnarray}
where
\begin{eqnarray}\label{gammadelta}
\gamma(\pmb{\alpha})=\frac{d}{\sum_{\substack{i,j=0\\i\ne j}}^{d-1}\frac{\alpha_i}{\alpha_j}},\quad \delta_k(\pmb{\alpha})=-\frac{\gamma(\pmb{\alpha})}{d}\sum_{\substack{i,j=0\\i\ne j}}^{d-1}\frac{\alpha_i}{\alpha_j}\omega^{k(d-j)}.
\end{eqnarray}
In our scenario, we consider that all the $N-1$ Bobs perform two arbitrary measurements denoted by $B_{i_0}$ and $B_{i_1}$. Also, we assume that the local dimensions of the state they share are unknown. At the same time, because Alice's device is trusted, her measurements are completely characterized and the local dimension of the shared state on Alice's side is known. Alice's observables are the generalized Pauli $Z$ and $X$ operators
\begin{align}
    Z_d=\sum_{i=0}^{d-1}\omega^i\ket{i}\bra{i},\qquad X_d=\sum_{i=0}^{d-1}\ket{i+1}\bra{i}.
\end{align} 
The corresponding steering inequality is given by
\begin{eqnarray}\label{steinsch}
    \langle \mathcal{S}_d(\pmb{\alpha},N)\rangle\leqslant\beta_L
\end{eqnarray}
where $\beta_L$ is the LHS bound given below.
\begin{fakt}
The local hidden state (LHS) bound $\beta_L$ of the steering functional $\langle  S_d(\pmb{\alpha})\rangle$ where $ S_d(\pmb{\alpha})$ is given in \eqref{Stefn1} is upper bounded by  
\begin{eqnarray}
    \beta_L\leqslant\max_{\substack{|\eta_0|,\ldots,|\eta_{d-1}|\\|\eta_0|^2+\ldots+|\eta_{d-1}|^2=1}}\left(d(N-1)\max_a{|\eta_a|^2}+\gamma(\pmb{\alpha})\left[\left(\sum_{i=0}^{d-1}|\eta_i|\right)^2-\sum_{i=0}^{d-1}\alpha_i\sum_{a=0}^{d-1}\frac{|\eta_a|^2}{\alpha_a}\right]\right)-(N-2).
\end{eqnarray}
Furthermore, $\beta_L<(d-1)(N-1)+1$.
\end{fakt}
\begin{proof}
In order to determine the LHS bound $\beta_L$ of the inequality in Eq. \eqref{Stefn1}, we first express the functional $\mathcal{S}_d(\pmb{\alpha},N)$ in the probability picture as following
\begin{align}
   S_d(\pmb{\alpha})=\ & d\sum_{j=1}^{N-1}\sum_{a,b_j=0}^{d-1}c_{a,b_j}p(a,b_j|0,0)+\gamma(\pmb{\alpha})\left(d\sum_{a,\mathbf{b}=0}^{d-1}c_{a,\mathbf{b}}p(a,\mathbf{b}|1,1,\ldots,1)-\sum_{\substack{i,a=0\\i\neq a}}^{d-1}\alpha_i\frac{p(a|0)}{\alpha_a}\right)\nonumber \\
   &-(N-1)-\gamma(\pmb{\alpha})-\delta_0(\pmb{\alpha}),
\end{align}
where $c_{a,b}=1$ if $a\oplus_d b=0$ and 0 otherwise, where $\oplus_d$ is the addition modulo $d$. Also, $c_{a,\mathbf{b}}=1$ if $a\oplus_d b_1\oplus_d \ldots \oplus_d b_{N-1}=0$ and 0 otherwise. We implement this notation because $\sum_{k=0}^{d-1}\omega^{kn(a-b)}=d\delta_{a,b}$ for any $n$. Likewise, $\sum_{k=0}^{d-1}\omega^{k(d+a-b)}=d\delta_{a,b}$. We note, from Eq. \eqref{gammadelta} that $\delta_0(\pmb{\alpha})=-1$, which results in
\begin{align}
   S_d(\pmb{\alpha})=\ &d\sum_{j=1}^{N-1}\sum_{a,b_j=0}^{d-1}c_{a,b_j}p(a,b_j|0,0)+\gamma(\pmb{\alpha})\left(d\sum_{a,\mathbf{b}=0}^{d-1}c_{a,\mathbf{b}}p(a,\mathbf{b}|1,1,\ldots,1)-\sum_{i=0}^{d-1}\alpha_i\sum_{a=0}^{d-1}\frac{p(a|0)}{\alpha_a}\right)-(N-2),
\end{align}

We consider the LHS model in which the correlations are given as
\begin{align}
   p(a,\mathbf{b}|x,\mathbf{y})=\sum_{\lambda\in \Lambda}p(\lambda)p(a|x,\rho^{\lambda}_A)p(\mathbf{b}|\mathbf{y},\lambda),  
\end{align}
where $\Lambda$ is some set of hidden variables $\lambda$ shared between Alice and all the Bobs with probability distribution $p(\lambda)$, whereas $p(a|x,\rho^{\lambda}_A)$ and $p(\mathbf{b}|\mathbf{y},\lambda)$ are local probability distributions. Recall that Alice's distributions are quantum and depend on the hidden state $\rho^{\lambda}_A$. The LHS model transforms our steering functional as
\begin{align}
   S_d(\pmb{\alpha})=&d\sum_{j=1}^{N-1}\sum_{a=0}^{d-1}\sum_{\lambda\in \Lambda}p(\lambda)p(a|0,\rho^{\lambda}_A)p(d-a|0,\lambda)+\gamma(\pmb{\alpha})\Biggl(d\sum_{a=0}^{d-1}\sum_{\lambda \in \Lambda}p(\lambda)p(a|1,\rho^{\lambda}_A)p(\mathbf{b}|1,\ldots,1,\lambda) \nonumber \\& -\sum_{i=0}^{d-1}\alpha_i\sum_{a=0}^{d-1}\frac{p(a|0,\rho_A)}{\alpha_a} \Biggr)-(N-2),  
\end{align}
where $b_j=d-a$ for every $j$ in the first term, $a\oplus_d b_1\oplus_d\ldots \oplus_d b_{N-1}=0$ in the second term, and in the last term we have used the fact that Alice's distributions are quantum and substituted $\rho_A=\sum_{\lambda}\rho^{\lambda}_A$. Now, we note that the first two terms in the above expression are bounded from above in the following way
\begin{align}
 \sum_{a=0}^{d-1}\sum_{\lambda\in \Lambda}p(\lambda)p(a|0,\rho^{\lambda}_A)p(d-a|0,\lambda)\leqslant   \sum_{\lambda\in \Lambda}p(\lambda)\max_a{p(a|0,\rho^{\lambda}_A)}\leqslant\max_{\ket{\psi}\in \mathcal{C}^d}\max_a{p(a|0,\ket{\psi})},\nonumber \\
 \sum_{a=0}^{d-1}\sum_{\lambda \in \Lambda}p(\lambda)p(a|1,\rho^{\lambda}_A)p(\mathbf{b}|1,\ldots,1,\lambda)\leqslant   \sum_{\lambda\in \Lambda}p(\lambda)\max_a{p(a|1,\rho^{\lambda}_A)}\leqslant\max_{\ket{\psi}\in \mathcal{C}^d}\max_a{p(a|1,\ket{\psi})},
\end{align}
where to get the first inequalities in both cases we used the fact that probabilities are normalized. To get the second inequality we exploit the fact that Alice's probability distributions are quantum and hence linear functions of the state. Substituting this in our function, we get
\begin{align}
    S_d(\pmb{\alpha})\leqslant\max_{\psi\in \mathcal{C}^d}\left[d(N-1)\max_a{p(a|0,\ket{\psi})}+\gamma(\pmb{\alpha})\left(d \max_ap(a|1,\ket{\psi})-\sum_{i=0}^{d-1}\alpha_i\sum_{a=0}^{d-1}\frac{p(a|0,\ket{\psi})}{\alpha_a}\right)\right]-(N-2),
\end{align}
where we have used the fact that the maximization over the state $\ket{\psi}$ is happening over the whole expression. We can further write the state $\ket{\psi}$ in the computational basis as $\ket{\psi}=\sum_i\eta_i\ket{i}$ and then substitute the explicit forms of Alice's observables $X$ and $Z$, which results in 
\begin{align}\label{upperbound_Schmidt}
    S_d(\pmb{\alpha})\leqslant\max_{\substack{|\eta_0|,\ldots,|\eta_{d-1}|\\|\eta_0|^2+\ldots+|\eta_{d-1}|^2=1}}\left(d(N-1)\max_a{|\eta_a|^2}+\gamma(\pmb{\alpha})\left[\left(\sum_{i=0}^{d-1}|\eta_i|\right)^2-\sum_{i=0}^{d-1}\alpha_i\sum_{a=0}^{d-1}\frac{|\eta_a|^2}{\alpha_a}\right]\right)-(N-2).
\end{align}
If we can now show that the terms inside the square bracket are less than zero, we can conclude that $S_d(\pmb{\alpha})\leqslant (d-1)(N-1)+1$. To do this, we notice the following
\begin{align}\label{Schwarz_ineq}
    \left(\sum_{i=0}^{d-1}|\eta_i|\right)^2=\left(\sum_{i=0}^{d-1}\sqrt{\alpha_i}\frac{|\eta_i|}{\sqrt{\alpha_i}}\right)^2\leqslant\sum_{i=0}^{d-1}\alpha_i\sum_{j=0}^{d-1}\frac{|\eta_j|^2}{\alpha_j},
\end{align}
where we have used the Cauchy-Schwarz inequality to get the last expression. Consequently, we have $S_d(\pmb{\alpha})\leqslant (d-1)(N-1)+1$. However, from Eq. \eqref{upperbound_Schmidt} we see that $S_d(\pmb{\alpha})= (d-1)(N-1)+1$ is possible iff $\max_a|\eta_a|^2=1$ as well as the terms inside the square brackets should also vanish. However, this is possible only if $\alpha_i=\lambda \eta_i$. Because $\sum_i\alpha_i^2=\sum_i|\eta_i|^2=1$, we must have $\lambda=1$. However, $\alpha_i<1$ for al $i$, which does not allow for $\max_a|\eta_a|^2=1$ to hold. Therefore, $S_d(\pmb{\alpha})= (d-1)(N-1)+1$ is never possible, which implies $S_d(\pmb{\alpha})< (d-1)(N-1)+1$.
\end{proof}
Let us now find the quantum bound of the steering functional $\langle \mathcal{S}_d(\pmb{\alpha},N)\rangle$.

\begin{fakt}
    The quantum bound of the steering functional $\langle \mathcal{S}_d(\pmb{\alpha},N)\rangle$ where $\mathcal{S}_d(\pmb{\alpha},N)$ is given in \eqref{Stefn1} is given by $\beta_Q=(d-1)(N-1)+1$ and is achieved by the Schmidt state $\ket{\psi(\alpha)}$ \eqref{Schmidt_SM} with the observables 
    \begin{equation}
    B_{j,0}=Z,\qquad  B_{j,1}=X.
\end{equation}
\end{fakt}
\begin{proof}
Let us rewrite our functional as \begin{align}\label{SandT}
    \mathcal{S}_d(\pmb{\alpha},N)=\sum_{k=1}^{d-1}\left(\sum_{j=1}^{N-1}A_0^k\otimes B_{j,0}^k\right)+T_{d}(\pmb{\alpha},N)
\end{align}
with
\begin{align}\label{T_expression}
    T_{d}(\pmb{\alpha},N)=\sum_{k=1}^{d-1}\left(\gamma(\pmb{\alpha}) A_1^{k}\otimes\bigotimes_{j=1}^{N-1} B_{j,1}^{k}+\delta_k(\pmb{\alpha})A_0^{k}\right).
\end{align}
Since, $A_0$ is unitary and $(B^{k}_{j|0})^{\dagger}B^{k}_{j|0}\leqslant \mathbbm{1}$ for any $k,j$, we have $|\bra{\psi_N}A_0^k\otimes B_{j,0}^k\ket{\psi_N}|\leqslant 1$ for all $k \in \{1,\ldots,d-1\}$, $j \in \{1,\ldots,N-1\}$ and $\ket{\psi_N}$. This implies that
\begin{align}\label{maxSandT}
    \max_{\ket{\psi_N}}\bra{\psi_N}\mathcal{S}_d(\pmb{\alpha},N)\ket{\psi_N}\leqslant (d-1)(N-1)+\max_{\ket{\psi_N}}\bra{\psi_N}T_d(\pmb{\alpha},N)\ket{\psi_N}
\end{align}
Therefore, we only need to prove that $\bra{\psi_N}T_d(\pmb{\alpha},N)\ket{\psi_N}\leqslant 1$. We can represent any pure state in the Hilbert space $\mathcal{H}_A\bigotimes_{i=1}^{N-1}\mathcal{H}_{B_i}$ as 
\begin{align}
    \ket{\psi_N}_{A\mathbf{B}}=\sum_{i=0}^{d-1}\lambda_i\ket{i}_A\ket{e_{i}}_{\mathbf{B}},
\end{align}
where $\lambda_i\geqslant 0$ and $\lambda_0^2+\ldots+\lambda_{d-1}^2=1$, and $\ket{e_i}$ are some vectors in the Hilbert space $\bigotimes_{i=1}^{N-1}\mathcal{H}_{B_j}$ such that $\ket{e_{i}}$ are not necessarily orthogonal. With this representation we can rewrite $\bra{\psi_N}T_d(\pmb{\alpha},N)\ket{\psi_N}$ as
\begin{eqnarray}\label{Schmidt_2ndpart}
    \bra{\psi_N}T_d(\pmb{\alpha},N)\ket{\psi_N}&=&\sum_{k=1}^{d-1}\sum_{i,j=0}^{d-1}\bigg[\gamma(\pmb{\alpha})\lambda_i\lambda_j\bra{i}A_1^k\ket{j}\bra{e_{i}}\bigotimes_{j=1}^{N-1} B_{j,1}^{k}\ket{e_{j}}+\delta_k(\pmb{\alpha})\lambda_i\lambda_j\bra{i}A_0^k\ket{j}\braket{e_{i}}{e_{j}}\bigg], \nonumber \\
    &=&\sum_{k=1}^{d-1}\sum_{i=0}^{d-1}\bigg[\gamma(\pmb{\alpha})\lambda_i\lambda_{i-k}\bra{e_{i}}\bigotimes_{j=1}^{N-1} B_{j,1}^{k}\ket{e_{j}}+\delta_k(\pmb{\alpha})\lambda_i^2\omega^{ik}\bigg],
\end{eqnarray}
where for the second equality we have substituted $A_i$ from \eqref{Aliceobservables}. Next, we note that 
\begin{align}
\sum_{k=1}^{d-1}\sum_{i=0}^{d-1}\delta_k(\pmb{\alpha})\lambda_i^2\omega^{ik}=1+\sum_{k=0}^{d-1}\sum_{i=0}^{d-1}\delta_k(\pmb{\alpha})\lambda_i^2\omega^{ik} 
=1+\gamma(\pmb{\alpha})-\gamma(\pmb{\alpha})\sum_{i,j=0}^{d-1}\frac{\alpha_i}{\alpha_j}\lambda_j^2.
\end{align}
where we get the first equality by using $\delta_0(\pmb{\alpha})=-1$ and the second equality by using the explicit form of $\delta_k(\pmb{\alpha})$.By substituting the above expression in Eq. \eqref{Schmidt_2ndpart}, we get
\begin{align}
    \bra{\psi_N}T_d(\pmb{\alpha},N)\ket{\psi_N}=1+\gamma(\pmb{\alpha})\sum_{k=0}^{d-1}\sum_{i=0}^{d-1}\lambda_i\lambda_{i-k}\text{Re}\left(\bra{e_{i}}\bigotimes_{j=1}^{N-1} B_{j,1}^{k}\ket{e_{j}}\right)-\gamma(\pmb{\alpha})\sum_{i,j=0}^{d-1}\frac{\alpha_i}{\alpha_j}\lambda_j^2,
\end{align}
since $\mathcal{S}_d$ is an Hermitian operator and every $\lambda_i$ is positive. In this expression we use the fact that Re($z$)$\leqslant |z|$ and  $(B^{k}_{j|0})^{\dagger}B^{k}_{j|0}=\mathbbm{1}$ for any $k,j$, to obtain
\begin{align}
    \bra{\psi_N}T_d(\pmb{\alpha},N)\ket{\psi_N} \leqslant 1+\gamma(\pmb{\alpha})\left[\left(\sum_{i=0}^{d-1}\lambda_i^2\right)^2-\sum_{i=0}^{d-1}\alpha_i\sum_{j=0}^{d-1}\frac{\lambda_j^2}{\alpha_j}\right].
\end{align}

We see that the term in square brackets has the same form as the terms in square brackets of Eq. \eqref{upperbound_Schmidt}, which is upper bounded by 0. This implies that 
\begin{align}\label{Tupper}
    \bra{\psi_N}T_d(\pmb{\alpha},N)\ket{\psi_N}\leqslant 1\implies \mathcal{S}_d(\pmb{\alpha},N)\leqslant (d-1)(N-1)+1.
\end{align}
It is now simple to observe that with the Schmidt state $\ket{\psi(\pmb{\alpha})}$ from \eqref{Schmidt_SM} and observables 
    \begin{equation}
    B_{j,0}=Z,\qquad  B_{j,1}=X,
\end{equation}
one can attain the value $\langle\mathcal{S}_d(\pmb{\alpha},N)\rangle=(d-1)(N-1)+1$. This completes the proof.
\end{proof}


Consequently, if all the parties attain the quantum bound $(d-1)(N-1)+1$ of $\mathcal{S}_d(\pmb{\alpha},N)$, the following conditions hold true
\begin{align}
    \bra{\psi_N}A_0^k\otimes B_{j,0}^k\ket{\psi_N}=1,\qquad
    \bra{\psi_N}T_d(\pmb{\alpha},N)\ket{\psi_N}=1,
\end{align}
for every $j \in \{1,\ldots,N-1\}$ and every $k \in \{1,\ldots,d-1\}$, where $T_d(\pmb{\alpha},N)$ is given in \eqref{T_expression}. These conditions hold because of Eqs. \eqref{SandT}, \eqref{maxSandT}, and \eqref{Tupper}. As a result, we have the following conditions for the state $\ket{\psi_N}$ and observables $A_i$ and $B_{j_i}$
\begin{align}
        A_0^k\otimes B_{j,0}^k\ket{\psi_N}=\ket{\psi_N},\label{stabilizersSchmidt1}\\ 
        \sum_{k=1}^{d-1}\left(\gamma(\pmb{\alpha}) A_1^{k}\substack{{N-1}\\ \bigotimes\\j=1} B_{j,1}^{k}+\delta_k(\pmb{\alpha})A_0^{k}\right)\ket{\psi_N}=\ket{\psi_N}\label{stabilizersSchmidt2}.
\end{align}
for every $j \in \{1,\ldots,N-1\}$ and every $k \in \{1,\ldots,d-1\}$. Since $A_i's$ and $B_i's$ are unitary observables, we can rewrite Eq. \eqref{stabilizersSchmidt1} as
\begin{eqnarray}\label{cond1}
\left(A_0^{(-k)}\otimes \I\right)\ket{\psi_N}=\left(\I\otimes B_{j,0}^{(k)}\right)\ket{\psi_N}
\end{eqnarray}
The measurements chosen by Alice are $A_0=Z$ and $A_1=X$. Thus, we have the following relation among Alice's observables 
\begin{eqnarray}\label{comm1}
A_1A_0=\omega A_0A_1
\end{eqnarray}
where $\omega=e^{\frac{2\pi i}{d}}$. As we are about to present, the above relations will be particularly useful for self-testing the Schmidt states.

\begin{thm}
Assume that the steering inequality \eqref{steinsch} is maximally violated when the $0^\text{th}$ party Alice chooses the observables $A_i$ defined to be $A_0=Z$, $A_1=X$ and all the other parties choose their observables as $B_{j,i}$ for $i\in \{0,1\}$ and $j=1,2\ldots,N-1$ acting $\mathcal{H}_{B_j}$. Let us say that the state which attains this violation is given by $\ket{\psi_N}_{A\mathbf{B}E}\in\mathbbm{C}^d\otimes\bigotimes_{i=1}^{N-1}\mathcal{H}_{B_i}\otimes\mathcal{H}_E$, then the following statement holds true for any $d$:  $\mathcal{H}_{B_i}=\mathcal{H}_{B'_i}\otimes\mathcal{H}_{B''_i}$, where $\mathcal{H}_{B'_j} \equiv \mathbbm{C}^d$, $\mathcal{H}_{B''_j}$ is some finite-dimensional Hilbert space, and there exists a local unitary transformation on Bob's side $U_j:\mathcal{H}_{B_j}\rightarrow\mathcal{H}_{B_j}$, such that
\begin{eqnarray}\label{Theo1.1sch}
\forall j, \quad U_j\,B_{j,0}\,U_j^{\dagger}=Z_{B'_j}\otimes\I_{B''_j},\qquad U_j\,B_{j,1}\,U_j^{\dagger}=X_{B'_j}\otimes\I_{B''_j}
\end{eqnarray}
where ${B''}$ denotes Bob's auxiliary system and the state   $\ket{\psi_{N}}$ is given by,
\begin{eqnarray}\label{Theo1.2sch}
\left(\I_{AE}\otimes \bigotimes_{i=1}^{N-1}U_j\right)\ket{\psi_N}_{A\mathbf{B}E}=\ket{\psi(\pmb{\alpha})}_{A\mathbf{B}'}
\otimes\ket{\xi}_{\mathbf{B}''E},
\end{eqnarray}
where $\ket{\xi}_{\mathbf{B}''E}\in\bigotimes_{j=1}^{N-1}\mathcal{H}_{B_j''}\otimes\mathcal{H}_E$ denotes the auxiliary state.
\end{thm}
\begin{proof}
We begin by taking \eqref{stabilizersSchmidt2} and multiplying it by $B_{l,1}\otimes\I$ on both sides,
\begin{eqnarray}\label{procond1}
\left(\sum_{k=1}^{d-1}\gamma(\pmb{\alpha}) A_1^{k}\otimes B_{l,1}^{k+1}\substack{{N-1}\\\bigotimes\\j=1\\j\ne l} B_{j,1}^{k}\right)\ket{\psi_N}=\left(\I-\sum_{k=1}^{d-1}\delta_k(\pmb{\alpha})A_0^{k}\right)\otimes B_{l,1}\ket{\psi_N}
\end{eqnarray}
Now, we multiply the above by $B_{l,0}\otimes\I$ on both the sides,
\begin{eqnarray}\label{procond2}
\left(\sum_{k=1}^{d-1}\delta_k(\pmb{\alpha}) A_1^{k}\otimes B_{l,0}B_{l,1}^{k+1}\substack{{N-1}\\\bigotimes\\j=1\\j\ne l} B_{j,1}^{k}\right)\ket{\psi_N}=\left(\I-\sum_{k=1}^{d-1}\delta_k(\pmb{\alpha})A_0^{k}\right)\otimes B_{l,0}B_{l,1}\ket{\psi_N}
\end{eqnarray}

Similarly, we multiply both the sides of \eqref{procond1} by $A_0^{-1}\otimes\I$ to have 
\begin{eqnarray}
\left(\sum_{k=1}^{d-1}\gamma(\pmb{\alpha}) A_0^{-1}A_1^{k}\otimes B_{l,1}^{k+1}\substack{{N-1}\\\bigotimes\\j=1\\j\ne l} B_{j,1}^{k}\right)\ket{\psi_N}=\left(\I-\sum_{k=1}^{d-1}\delta_k(\pmb{\alpha})A_0^{k}\right)A_0^{-1}\otimes B_{l,1}\ket{\psi_N}.
\end{eqnarray}
By using the commutation relation \eqref{comm1}, and since $B_{j,1}\otimes\I$ and $A_0\otimes\I$ commute, we get
\begin{eqnarray}
\left(\sum_{k=1}^{d-1}\gamma(\pmb{\alpha}) \omega^{k}A_1^{k}\otimes B_{l,1}^{k+1}\substack{{N-1}\\\bigotimes\\j=1\\j\ne l} B_{j,1}^{k}\right)A_0^{-1}\ket{\psi_N}=\left(\I-\sum_{k=1}^{d-1}\delta_k(\pmb{\alpha})A_0^{k}\right)A_0^{-1}\otimes B_{l,1}\ket{\psi_N}
\end{eqnarray}
In the above equation we use \eqref{cond1} for $j=l$, which gives
\begin{eqnarray}\label{procond3}
\left(\sum_{k=1}^{d-1}\gamma(\pmb{\alpha}) \omega^{k}A_1^{k}\otimes B_{l,1}^{k+1}B_{l,0}\substack{{N-1}\\\bigotimes\\j=1\\j\ne l} B_{j,1}^{k}\right)\ket{\psi_N}=\left(\I-\sum_{k=1}^{d-1}\delta_k(\pmb{\alpha})A_0^{k}\right)\otimes B_{l,1}B_{l,0}\ket{\psi_N}.
\end{eqnarray}
We can proceed now by multiplying \eqref{procond3} by $\omega$ and subtracting it from \eqref{procond2} as follows
\begin{eqnarray}\label{procond4}
\left(\sum_{k=1}^{d-1}\gamma(\pmb{\alpha}) A_1^{k}\otimes (B_{l,0}B_{l,1}^{k+1}-\omega^{k+1}B_{l,1}^{k+1}B_{l,0})\substack{{N-1}\\\bigotimes\\j=1\\j\ne l} B_{j,1}^{k}\right)\ket{\psi_N}=\left(\I-\sum_{k=1}^{d-1}\delta_k(\pmb{\alpha})A_0^{k}\right)\otimes (B_{l,0}B_{l,1}-\omega B_{l,1}B_{l,0})\ket{\psi_N}
\end{eqnarray}
Next, we again take \eqref{stabilizersSchmidt2} and multiply it by $A_1^{-1}\otimes\I$ on both sides,
\begin{eqnarray}\label{procond1_}
\left(\sum_{k=1}^{d-1}\gamma(\pmb{\alpha}) A_1^{k-1}\otimes B_{l,1}^{k}\substack{{N-1}\\\bigotimes\\j=1\\j\ne l} B_{j,1}^{k}-\delta_k(\pmb{\alpha})A_1^{-1}A_0^{k}\right)\ket{\psi_N}=A_1^{-1}\ket{\psi_N}.
\end{eqnarray}
We multiply the above by $B_{l,0}\otimes\I$, use the fact that $B_{l,0}\otimes\I$ commutes with $A_1\otimes\I$, and then use \eqref{cond1} for $k=1$ to obtain
\begin{eqnarray}\label{procond5}
\left(\sum_{k=1}^{d-1}\gamma(\pmb{\alpha}) A_1^{k-1}\otimes B_{l,0}B_{l,1}^{k}\substack{{N-1}\\\bigotimes\\j=1\\j\ne l} B_{j,1}^{k}-\delta_k(\pmb{\alpha})A_1^{-1}A_0^{k-1}\right)\ket{\psi_N}=A_1^{-1}A_0^{-1}\ket{\psi_N}.
\end{eqnarray}
Now, we multiply by $A_0^{-1}\otimes\I$ on both the sides of \eqref{procond1_}
\begin{eqnarray}
\left(\sum_{k=1}^{d-1}\gamma(\pmb{\alpha}) A_0^{-1}A_1^{k-1}\otimes B_{l,1}^{k}\substack{{N-1}\\\bigotimes\\j=1\\j\ne l} B_{j,1}^{k}-\delta_k(\pmb{\alpha})A_0^{-1}A_1^{-1}A_0^{k}\right)\ket{\psi_N}=A_0^{-1}A_1^{-1}\ket{\psi_N}
\end{eqnarray}
Using the commutation relation \eqref{comm1}, the condition \eqref{cond1}, and the fact that $B_{j,1}\otimes\I$ and $A_0\otimes\I$ commute, we obtain that
\begin{eqnarray}\label{procond6}
\left(\sum_{k=1}^{d-1}\gamma(\pmb{\alpha})\omega^{k-1}A_1^{k-1}\otimes B_{l,1}^{k}B_{l,0}\substack{{N-1}\\\bigotimes\\j=1\\j\ne l} B_{j,1}^{k}-\omega^{-1} \delta_k(\pmb{\alpha})A_1^{-1}A_0^{k-1}\right)\ket{\psi_N}=\omega^{-1} A_1^{-1}A_0^{-1}\ket{\psi_N}.
\end{eqnarray}
We proceed now by multiplying \eqref{procond5} by $\omega$ and subtracting it from \eqref{procond6} to obtain 
\begin{eqnarray}
\left(\sum_{k=1}^{d-1}\gamma(\pmb{\alpha}) A_1^{k-1}\otimes (B_{l,0}B_{l,1}^{k}-\omega^{k}B_{l,1}^{k}B_{l,0})\substack{{N-1}\\\bigotimes\\j=1\\j\ne l} B_{j,1}^{k}\right)\ket{\psi_N}= 0,
\end{eqnarray}
which can be broken up as 
\begin{eqnarray}
\left(\sum_{k=2}^{d-1}\gamma(\pmb{\alpha}) A_1^{k-1}\otimes (B_{l,0}B_{l,1}^{k}-\omega^{k}B_{l,1}^{k}B_{l,0})\substack{{N-1}\\\bigotimes\\j=1\\j\ne l} B_{j,1}^{k}\right)\ket{\psi_N}= -\gamma(\pmb{\alpha})\left( (B_{l,0}B_{l,1}-\omega B_{l,1}B_{l,0})\substack{{N-1}\\\bigotimes\\j=1\\j\ne l} B_{j,1}\right)\ket{\psi_N}.
\end{eqnarray}
The above can be further simplified as
\begin{eqnarray}\label{procond7}
\left(\sum_{k=2}^{d-1}\gamma(\pmb{\alpha})A_1^{k-1}\otimes (B_{l,0}B_{l,1}^{k}-\omega^{k}B_{l,1}^{k}B_{l,0})\substack{{N-1}\\\bigotimes\\j=1\\j\ne l} B_{j,1}^{k-1}\right)\ket{\psi_N}= -\gamma(\pmb{\alpha})\left( B_{l,0}B_{l,1}-\omega B_{l,1}B_{l,0}\right)\ket{\psi_N}.
\end{eqnarray}
Now, we can see that the l.h.s. of \eqref{procond4} and \eqref{procond7} are the same. Thus, we can conclude that
\begin{eqnarray}
\left(\I-\sum_{k=1}^{d-1}\delta_k(\pmb{\alpha})A_0^{k}\right)\otimes (B_{l,0}B_{l,1}-\omega B_{l,1}B_{l,0})\ket{\psi_N}=-\gamma(\pmb{\alpha})\left( B_{l,0}B_{l,1}-\omega B_{l,1}B_{l,0}\right)\ket{\psi_N}
\end{eqnarray}
We finally get
\begin{eqnarray}
\left(\I(1+\gamma(\pmb{\alpha}))-\sum_{k=1}^{d-1}\delta_k(\pmb{\alpha})A_0^{k}\right)\otimes (B_{l,0}B_{l,1}-\omega B_{l,1}B_{l,0})\ket{\psi_N}=0.
\end{eqnarray}
It can be shown that if $(P\otimes Q ).v=0$ for any $v\neq 0$ and $P$ is invertible, then we can conclude that $Q=0$. We can use this fact to show that $\left(\I(1+\gamma(\pmb{\alpha}))-\sum_{k=1}^{d-1}\delta_k(\pmb{\alpha})A_0^{k}\right)$ is invertible. Indeed, by expanding it term by term, we get
\begin{eqnarray}
\left(\I(1+\gamma(\pmb{\alpha}))-\sum_{k=1}^{d-1}\delta_k(\pmb{\alpha})A_0^{k}\right)=\sum_{l=0}^{d-1}\left((1+\gamma(\pmb{\alpha}))-\sum_{k=1}^{d-1}\delta_k(\pmb{\alpha})\omega^{kl}\right)\ket{l}\bra{l}.
\end{eqnarray}
Notice that the above matrix is not invertible if any of the elements is 0. Thus, we solve the r.h.s. of the above equation for the $l^\text{th}$ element
\begin{align}
(1+\gamma(\pmb{\alpha}))+\frac{\gamma(\pmb{\alpha})}{d}\left(\sum_{k=1}^{d-1}\sum_{\substack{i,j=0\\i\ne j}}^{d-1}\frac{\alpha_i}{\alpha_j}\omega^{k(l-j)}\right)&=(1+\gamma(\pmb{\alpha}))+\frac{\gamma(\pmb{\alpha})}{d}\left(\sum_{\substack{i=0\\i\ne l}}^{d-1}\frac{\alpha_i}{\alpha_l}(d-1)-\sum_{\substack{i,j=0\\i\ne j,j\ne l}}^{d-1}\frac{\alpha_i}{\alpha_j}\right)\nonumber\\
&=(1+\gamma(\pmb{\alpha}))+\gamma(\pmb{\alpha})\left(\sum_{\substack{i=0}}^{d-1}\frac{\alpha_i}{\alpha_l}-\frac{1}{d}\sum_{\substack{i,j=0}}^{d-1}\frac{\alpha_i}{\alpha_j}\right)
\end{align}
where we have used the fact that $\sum_{k=0}^{d-1}\omega^{k(l-j)}=d\delta_{l,j}$. Substituting $\gamma(\pmb{\alpha})$ from \eqref{gammadelta}, we get
\begin{align}
(1+\gamma(\pmb{\alpha}))+\gamma(\pmb{\alpha})\left(\sum_{\substack{i=0}}^{d-1}\frac{\alpha_i}{\alpha_l}-\frac{1}{d}\sum_{\substack{i,j=0}}^{d-1}\frac{\alpha_i}{\alpha_j}\right)&=\gamma(\pmb{\alpha})\left(1+\frac{1}{d}\sum_{\substack{i,j=0\\i\ne j}}^{d-1}\frac{\alpha_i}{\alpha_j}-\frac{1}{d}\sum_{\substack{i,j=0}}^{d-1}\frac{\alpha_i}{\alpha_j}+\sum_{\substack{i=0}}^{d-1}\frac{\alpha_i}{\alpha_l}\right),\nonumber\\
&=\gamma(\pmb{\alpha})\left(\frac{1}{d}\sum_{\substack{i,j=0}}^{d-1}\frac{\alpha_i}{\alpha_j}-\frac{1}{d}\sum_{\substack{i,j=0}}^{d-1}\frac{\alpha_i}{\alpha_j}+\sum_{\substack{i=0}}^{d-1}\frac{\alpha_i}{\alpha_l}\right)
=\gamma(\pmb{\alpha})\left(\sum_{\substack{i=0}}^{d-1}\frac{\alpha_i}{\alpha_l}\right)>0.
\end{align}
Thus, we can conclude that the matrix $\left(\I(1+\gamma(\pmb{\alpha}))-\sum_{k=1}^{d-1}\delta_k(\pmb{\alpha})A_0^{k}\right)$ is invertible. Therefore, we have 
\begin{eqnarray}
\I\otimes (B_{l,0}B_{l,1}-\omega B_{l,1}B_{l,0})\ket{\psi_N}=0
\end{eqnarray}
Finally, we can conclude that $B_{l,0}B_{l,1}=\omega B_{l,1}B_{l,0}$. As was shown in \cite{Jed1}, if two observables $B_0$ and $B_1$ follow the above relation, then there exists a unitary matrix $U$ such that
\begin{eqnarray}
UB_{l,0}U^{\dagger}= Z^{\dagger}_{B_l'}\otimes\I_{B_l''},\quad UB_{l,1}U^{\dagger}=X_{B_l'}\otimes\I_{B_l''}.
\end{eqnarray}
From the above relations, we can use \eqref{stabilizersSchmidt1} and \eqref{stabilizersSchmidt2} to show that the state which saturates the quantum bound is the Schmidt state of local dimension $d$.

Now, let us prove that the state maximally violating the steering inequality must be equivalent to the Schmidt state, up to local unitary transformations. By taking a general state of the form $\ket{\psi_N}=\sum_{i,j_1,j_2,\ldots,j_{N-1}=0}^{d-1}\ket{ij_1j_2\ldots j_{N-1}}\ket{\psi_{ij_1j_2\ldots j_{N-1}}}$, we can plug them in \eqref{stabilizersSchmidt1} with $k=1$ for the $l^\text{th}$ Bob to obtain
\begin{eqnarray}
(Z\otimes Z^{\dagger}_{l}\otimes\I)\left(\sum_{i,j_1,j_2,\ldots,j_{N-1}=0}^{d-1}\ket{ij_1j_2\ldots j_{N-1}}\ket{\psi_{ij_1j_2\ldots j_{N-1}}}\right)=\sum_{i,j_1,j_2,\ldots,j_{N-1}=0}^{d-1}\ket{ij_1j_2\ldots j_{N-1}}\ket{\psi_{ij_1j_2\ldots j_{N-1}}},
\end{eqnarray}
which is equivalent to,
\begin{eqnarray}
\sum_{i,j_1,j_2,\ldots,j_{N-1}=0}^{d-1}\omega^{i-j_l}\ket{ij_1j_2\ldots j_{N-1}}\ket{\psi_{ij_1j_2\ldots j_{N-1}}}=\sum_{i,j_1,j_2,\ldots,j_{N-1}=0}^{d-1}\ket{ij_1j_2\ldots j_{N-1}}\ket{\psi_{ij_1j_2\ldots j_{N-1}}}.
\end{eqnarray}
For $i\neq j_l$, the above identity implies $\omega^{j_l-i}\ket{\psi_{ij_1j_2\ldots j_{N-1}}}=\ket{\psi_{ij_1j_2\ldots j_{N-1}}}$. This indicates that
\be \label{sf1}
\forall i,j_l, \text{ such that } i\neq j_l, \quad \ket{\psi_{ij_1j_2\ldots j_{N-1}}}=0. 
\ee 
Consequently, we must have $i=j_l$ for all $l$. Next, we can proceed similarly with Eq. \eqref{stabilizersSchmidt2} to obtain
\begin{eqnarray}
\left(\sum_{k=1}^{d-1}\gamma(\pmb{\alpha}) X^{k}\substack{{N-1}\\\bigotimes\\j=1} X^{k}\otimes\mathbbm{1}+\delta_k(\pmb{\alpha})Z^{k}\right)\left(\sum_{i=0}^{d-1}\ket{ii\ldots i}\ket{\psi_{ii\ldots i}}\right)=\sum_{i=0}^{d-1}\ket{ii\ldots i}\ket{\psi_{ii\ldots i}}.
\end{eqnarray}
After performing the operations, we get
\begin{eqnarray}
\sum_{i=0}^{d-1}\left(\sum_{k=1}^{d-1}\gamma(\pmb{\alpha})\ket{i+k}\ket{i+k}\ldots\ket{i+k}+\omega^{ki}\delta_k(\pmb{\alpha})\ket{ii\ldots i}\right)\ket{\psi_{ii\ldots i}}=\sum_{i=0}^{d-1}\ket{ii\ldots i}\ket{\psi_{ii\ldots i}}.
\end{eqnarray}
If we take the inner product of the above relation with $\bra{ss\ldots s}$, we obtain
\begin{eqnarray}
\sum_{k=1}^{d-1}\left(\gamma(\pmb{\alpha})\ket{\psi_{s-k,s-k\ldots ,s-k}}+\omega^{ks}\delta_k(\pmb{\alpha})\right)\ket{\psi_{ss\ldots s}}=\ket{\psi_{ss\ldots s}}.
\end{eqnarray}
From \eqref{gammadelta} we have that
\begin{eqnarray}
\sum_{k=1}^{d-1}\left(\ket{\psi_{s-k,s-k\ldots ,s-k}}-\frac{\omega^{ks}}{d}\sum_{\substack{i,j=0\\i\ne j}}^{d-1}\frac{\alpha_i}{\alpha_j}\omega^{k(d-j)}\right)\ket{\psi_{ss\ldots s}}=\frac{1}{d}\sum_{\substack{i,j=0\\i\ne j}}^{d-1}\frac{\alpha_i}{\alpha_j}\ket{\psi_{ss\ldots s}},
\end{eqnarray}
which can be simplified in the following way
\begin{eqnarray}
\sum_{k=1}^{d-1}\ket{\psi_{s-k,s-k\ldots ,s-k}}=\frac{\sum_{i=0}^{d-1}\alpha_i}{\alpha_s}\ket{\psi_{ss\ldots s}}.
\end{eqnarray}
Thus, for any $s$ we have
\begin{eqnarray}
\ket{\psi_{ss\ldots s}}=\frac{\sum_{k=1}^{d-1}\ket{\psi_{s-k,s-k\ldots ,s-k}}}{\sum_{i=0}^{d-1}\alpha_i}\alpha_s,\quad\forall s\in\{0,1,\ldots,d-1\}.
\end{eqnarray}
If we define
\be
\ket{\xi}_{\mathbf{B}''E}\coloneqq \frac{\sum_{k=1}^{d-1}\ket{\psi_{s-k,s-k\ldots ,s-k}}}{\sum_{i=0}^{d-1}\alpha_i},
\ee
then we can write
\begin{eqnarray}
\left(\I_{AE}\otimes\bigotimes_{i=1}^{N-1}U_j\right)\ket{\psi_N}_{A\mathbf{B}E}=\left(\sum_{i=0}^{d-1}\alpha_{i}\ket{i}^{\otimes N}\right)\otimes\ket{\xi}_{\mathbf{B}''E},
\end{eqnarray}
which completes the proof.
\end{proof}

\section{Self-test $N$-qubit generalized $W$-States}\label{AppendixC}

The $N$-qubit generalized $W$ states have the following form
\begin{eqnarray}\label{WN}
\ket{\psi_{W_N}}=\sum_{l=0}^{N-1}\alpha_{l+1}\ket{0\ldots,0,1_l,0,\ldots,0}=\alpha_{1}\ket{10\ldots 00}+\alpha_{2}\ket{01\ldots 00}+\ldots +\alpha_N\ket{00\ldots 01},
\end{eqnarray}
such that $\alpha_i>0$ for every $i$ and $\sum_{i=1}^{N}\alpha_i^2=1$. We introduce a steering inequality that is maximally violated by the above state through the following steering operator
\begin{align}\label{steinWN}
\mathcal{W}_N(\pmb{\alpha})  = -2Z_A\otimes \bigotimes_{k=1}^{N-1}B_{k,0}&+\sum_{l=1}^{N-1}  \left[ Z_A \otimes \I_l\otimes\left(\I-P_l\right)    +\left(\gamma_l X_A \otimes B_{l,1}+\delta_lZ_A \otimes \I_{l}\right)\otimes P_l \right]\nonumber\\
&+\sum_{l=1}^{N-1}  \left[ \I_A \otimes B_{l,0}\otimes\left(\I-P_l\right)    +\I_A \otimes \I_l\otimes P_l \right],
\end{align}
where
\begin{align}\label{coeffN}
    \gamma_l=\frac{2\alpha_{l+1}\alpha_1}{\alpha_{l+1}^2+\alpha_1^2},\quad \delta_l=\frac{\alpha_{l+1}^2-\alpha_1^2}{\alpha_{l+1}^2+\alpha_1^2},\qquad \forall l\in \{1,\ldots,N-1\},
\end{align}
with 
\begin{eqnarray}\label{Pl}
   P_l=\frac{1}{2^{N-2}} \bigotimes_{k=1, k \neq l}^{N-1}\left(\I_k+ B_{k,0}\right)
\end{eqnarray}
and $B_{k,0}$ and $B_{k,1}$ are the measurements acting on the $k$-th qubit. In the same way, $\I_l$ is the identity operation acting on the $l$-th qubit. The corresponding steering inequality is given by
\begin{eqnarray}\label{inequalityWN}
    \langle \mathcal{W}_N(\pmb{\alpha})\rangle\leqslant\beta_L.
\end{eqnarray}
where $\beta_L$ is the LHS bound.

Let us now find the quantum bound of the steering functional $\langle \mathcal{W}_N(\pmb{\alpha})\rangle$.

\begin{fakt}\label{Qbound_Wstate}
    The quantum bound of the steering functional $\langle \mathcal{W}_N(\pmb{\alpha})\rangle$ where $\mathcal{W}_N(\pmb{\alpha})$ is given in \eqref{steinWN} is given by $\beta_Q=2N$ and is achieved by the $N$-qubit generalized $W$ state $\ket{\psi_{W_N}}$ \eqref{WN} with the observables 
    \begin{equation}
    B_{j,0}=Z,\qquad  B_{j,1}=X.
\end{equation}
\end{fakt}
\begin{proof}
    Before proceeding, let us recall here that $B_{j,0}$ are unitary observables as they correspond to projective measurements. Now, let us observe that the first term in the steering operator $\mathcal{W}_N$ \eqref{steinWN} is  
\begin{eqnarray}
   Z_A\otimes \bigotimes_{k=1}^{N-1}B_{k,0}\leqslant\I
\end{eqnarray}
where we used the fact that $| B_{k,0}|\leqslant\I$.
Similarly, one can straightforwardly conclude that for any $l$
\begin{eqnarray}
    \I_A \otimes B_{l,0}\otimes\left(\I-P_l\right)    +\I_A \otimes \I_l\otimes P_l\leqslant\I
\end{eqnarray}
Let us now observe that for any $l$, $B_{l,1}=\Pi_l^+-\Pi_l^-$ where $\Pi_l^+,\Pi_l^-$ denote the projector onto the positive and negative eigenspace respectively. Now, we observe that  
\begin{eqnarray}
  |\gamma_l X_A \otimes B_{l,1}+\delta_lZ_A \otimes \I_{l}|\leqslant   |(\gamma_l X_A +\delta_lZ_A) \otimes \Pi_{l}^+| +|(-\gamma_l X_A +\delta_lZ_A) \otimes \Pi_{l}^-|.
\end{eqnarray}
As $\Pi_l^+,\Pi_l^-$ are positive matrices, we can conclude from the above formula that
\begin{eqnarray}
    | \gamma_l X_A \otimes B_{l,1}+\delta_lZ_A \otimes \I_{l}|\leqslant \lambda^{(1)}_{\mathrm{max}}\I_A\otimes\Pi_{l}^+ + \lambda^{(2)}_{\mathrm{max}}\I_A\otimes\Pi_{l}^-
\end{eqnarray}
where $\lambda^{(1)}_{\mathrm{max}},\lambda^{(2)}_{\mathrm{max}}$ are the maximum eigenvalues of $|\gamma_l X_A +\delta_lZ_A|,|-\gamma_l X_A +\delta_lZ_A|$ respectively. One can straightaway evaluate that $\lambda^{(1)}_{\mathrm{max}}=\lambda^{(2)}_{\mathrm{max}}=1$ and thus we obtain that
\begin{eqnarray}
     |\gamma_l X_A \otimes B_{l,1}+\delta_lZ_A \otimes \I_{l}|\leqslant \I_{A,l}
\end{eqnarray}
Finally considering the second term in \eqref{steinWN} for any $l$, we have that
\begin{eqnarray}
    Z_A \otimes \I_l\otimes\left(\I-P_l\right)    +\left(\gamma_l X_A \otimes B_{l,1}+\delta_lZ_A \otimes \I_{l}\right)\otimes P_l\leqslant \I_{A,l}\otimes\left(\I-P_l\right)+ |\gamma_l X_A \otimes B_{l,1}+\delta_lZ_A \otimes \I_{l}|\otimes P_l\leqslant\I.
\end{eqnarray}
Thus, the maximum quantum value of the steering inequality \eqref{steinWN} is $2N$ and can be achieved when $B_{j,0}=Z$ and $B_{j,1}=X$ with the generalised $W$ state $\ket{\psi_{W_N}}$.
\end{proof}

To achieve the maximal violation of a steering inequality through operator \eqref{steinWN}, the following  conditions need to be satisfied:
\begin{eqnarray} \label{rel,1WN}
    Z_A\otimes \bigotimes_{k=1}^{N-1}B_{k,0}\ket{\psi}_{A\mathbf{B}}=-\ket{\psi}_{A\mathbf{B}}
\end{eqnarray}
and for all $l=1,\ldots,N-1$,
\begin{eqnarray}\label{rellWN}
    \left[ Z_A \otimes  \I_l\left(\I-P_l\right)+\left(\gamma_l X_A \otimes B_{l,1}
+\delta_lZ_A \otimes \I_{l}\right)\otimes P_l  \right] \ket{\psi}_{A\mathbf{B}}=\ket{\psi}_{A\mathbf{B}}
\end{eqnarray}
and,
\begin{eqnarray}\label{rellWNnew}
   \left[ \I_A \otimes B_{l,0}\otimes\left(\I-P_l\right)    +\I_A \otimes \I_l\otimes P_l \right]\ket{\psi}_{A\mathbf{B}}=\ket{\psi}_{A\mathbf{B}}.
\end{eqnarray}

Now, let us proceed to the self-testing statement.
\begin{thm}
Assume that the steering inequality \eqref{inequalityWN} is maximally violated when the $0^\text{th}$ party Alice chooses the observables $A_i$ defined to be $A_0=Z$, $A_1=X$ and all the other parties choose their observables as $B_{j,i}$ for $i\in \{0,1\}$ and $j=1,2\ldots,N-1$ acting $\mathcal{H}_{B_j}$. Let us say that the state which attains this violation is given by $\ket{\psi_N}_{A\mathbf{B}E}\in\mathbbm{C}^2\otimes\bigotimes_{j=1}^{N-1}\mathcal{H}_{B_j}\otimes\mathcal{H}_E$, then the following statement holds true:  $\mathcal{H}_{B_j}=\mathcal{H}_{B'_j}\otimes\mathcal{H}_{B''_j}$, where $\mathcal{H}_{B'_j} \equiv \mathbbm{C}^2$, $\mathcal{H}_{B''_j}$ is some finite-dimensional Hilbert space, and there exists a local unitary transformation on Bob's side $U_j:\mathcal{H}_{B_j}\rightarrow\mathcal{H}_{B_j}$, such that
\begin{eqnarray}\label{Theo3_measurments}
\forall j, \quad U_j\,B_{j,0}\,U_j^{\dagger}=Z_{B'_j}\otimes\I_{B''_j},\qquad U_j\,B_{j,1}\,U_j^{\dagger}=X_{B'_j}\otimes\I_{B''_j}
\end{eqnarray}
where ${B''}$ denotes Bob's auxiliary system and the state   $\ket{\psi_{N}}$ is given by,
\begin{eqnarray}\label{Theo3_states}
\left(\I_{AE}\otimes \bigotimes_{i=1}^{N-1}U_j\right)\ket{\psi_N}_{A\mathbf{B}E}=\ket{\psi_{W_N}}_{A\mathbf{B}'}
\otimes\ket{\xi}_{\mathbf{B}''E},
\end{eqnarray}
where $\ket{\xi}_{\mathbf{B}''E}\in\bigotimes_{j=1}^{N-1}\mathcal{H}_{B_j''}\otimes\mathcal{H}_E$ denotes the auxiliary state.
\end{thm}

\begin{proof}
For simplicity, we drop the lower indices from the state $\ket{\psi}_{A\mathbf{B}E}$. First, let us consider \eqref{rel,1WN} and rewrite it as
\begin{eqnarray}
      Z_A\otimes B_{l,0}\otimes \bigotimes_{k=1,k\ne l}^{N-1}B_{k,0}\ket{\psi}=-\ket{\psi}
\end{eqnarray}
and multiply it by $B_{l',0}$ for $l\ne l'$ and use the fact that $B_{l',0}^2=\I$. By adding the result to \eqref{rel,1WN} we obtain
\be
    Z_A\otimes B_{l,0} \otimes(\I+B_{l',0})\otimes\bigotimes_{k=1, k\neq l,l'}^{N-1}B_{k,0}\ket{\psi}=-\I \otimes  (\I+B_{l',0})\ket{\psi}.\label{ZBL0psi}
\ee
Then, multiplying \eqref{ZBL0psi} with $B_{l'',0}$ and then adding the resulting formula with \eqref{ZBL0psi}, we obtain
\begin{eqnarray}
     Z_A\otimes B_{l,0} \otimes(\I+B_{l',0})\otimes(\I+B_{l'',0})\otimes\bigotimes_{k=1, k\neq l,l',l''}^{N-1}B_{k,0}\ket{\psi}=-\I \otimes  (\I+B_{l',0})\otimes  (\I+B_{l'',0})\ket{\psi}.
\end{eqnarray}
Continuing in a similar manner, we get that
\begin{eqnarray}\label{ZBL0psi1}
     Z_A\otimes B_{l,0} \otimes P_l\ket{\psi}=-\I \otimes  P_l\ket{\psi}
\end{eqnarray}
where $P_l$ is given in Eq. \eqref{Pl}. Denoting $\ket{\Tilde{\psi}}\coloneqq P_l\ket{\psi}$, we have 
\begin{eqnarray}\label{ZBL0psi2}
     Z_A\otimes B_{l,0}\ket{\tilde{\psi}}=-\ket{\tilde{\psi}}.
\end{eqnarray}
Next, we consider \eqref{rellWN} for any $l$ and multiply it by $P_l$, which using the fact that $P_l^2=P_l$ gives us 
\begin{align}
\left(\gamma_l X_A \otimes B_{l,1}
+\delta_lZ_A \otimes \I_l\right)\ket{\tilde{\psi}}= \ket{\tilde{\psi}}
\end{align}
which, using the fact that $B_{l,1}$ is unitary, can be rearranged to
\begin{eqnarray}\label{gammaXx1psi}
  \gamma_l X_A \ket{\tilde{\psi}}=\left(\I- \delta_lZ_A \otimes \I_l\right)\otimes B_{l,1}\ket{\tilde{\psi}}
\end{eqnarray}
Now, let us multiply the above by $Z_A\otimes \I_l$ and use the fact that this term commutes with $\left(\I_A-\delta_lZ_A\right)\otimes B_{l,1}$. That gives us
\be
\gamma_l Z_A X_A\otimes \I_l\ket{\Tilde{\psi}}=\left[\left(\I_A-\delta_lZ_A\right)\otimes B_{l,1}\right] (Z_A\otimes\I_l)\ket{\Tilde{\psi}}
\ee
which utilising \eqref{ZBL0psi2} gives us
\be\label{ZX=B1B0}
\gamma_l Z_A X_A\otimes \I_l\ket{\Tilde{\psi}}=-\left(\I_A-\delta_lZ_A\right)\otimes B_{l,1}B_{l,0}\ket{\Tilde{\psi}}.
\ee
Now, let us consider \eqref{ZBL0psi2} and multiply it by $ \gamma_lX_A\otimes \I_l$. Similarly, we obtain
\be
\gamma_l X_AZ_A\otimes \I_l\ket{\Tilde{\psi}}=-(\I_A\otimes B_{l,0})(\gamma_l X_A\otimes \I_l)\ket{\Tilde{\psi}}.
\ee
Using\eqref{gammaXx1psi}, the above formula can be expressed as
\be\label{XZ=B0B1}
\gamma_l X_AZ_A\otimes \I_l\ket{\Tilde{\psi}}=-\left(\I_A-\delta_lZ_A\right)\otimes B_{l,0}B_{l,1} \ket{\Tilde{\psi}}.
\ee
By adding \eqref{XZ=B0B1} to \eqref{ZX=B1B0} and taking into account that $ZX+XZ=0$, we obtain
\be
(\I_A-\delta_lZ_A)\otimes(B_{l,1}B_{l,0}+B_{l,0}B_{l,1})\ket{\Tilde{\psi}}=0.
\ee
Because $\alpha_i\neq 0$ for any $i$, we have that $\delta_l\neq 1$ for any $l$. Therefore, $\I_A-\delta_lZ_A$ is invertible, which then using the fact that the local states are full-rank allows us to conclude from the above formula that
\be
\{B_{l,0},B_{l,1}\}=0.
\ee
Since the operators $B_{l,0}$ and $B_{l,1}$ anti-commute, we can use the result from \cite{Jed1} to guarantee that there exist unitaries $U_{B_l}$ for every $l=1,\ldots,N-1$ such that 
\begin{align}\label{Bl}
    U_{B_l}B_{l,0}U_{B_l}^{\dagger}=Z\otimes\I,\qquad U_{B_l}B_{l,1}U_{B_l}^{\dagger}=X\otimes\I.
\end{align}

Let us now find the state that maximally violates the steering inequality provided by the operator \eqref{steinWN}. As the Hilbert spaces of all parties decompose as $\mathbb{C}^2\otimes H_{B_j''}$, we consider a state in a general form as
\begin{align}\label{genstateN}
    \bigotimes_{l=1}^{N-1}U_{B_l}\ket{\psi}=\sum_{i_0,i_1,\ldots, i_{N-1}=0,1}\ket{i_0i_1\ldots i_{N-1}}\ket{\psi_{i_0i_1\ldots i_{N-1}}},
\end{align}
where $U_{B_l}$ are the unitaries in \eqref{Bl} and $\ket{\psi_{i_0i_1\ldots i_{N-1}}}$ are some unnormalized states. We can multiply \eqref{rel,1WN} by $\bigotimes_{l=1}^{N-1} U_{B_l}$ to obtain
\be
Z_A\otimes \bigotimes_{l=1}^{N-1}\left(U_{B_l}B_{l,0}U_{B_l}^\dagger\right) U_{B_l}\ket{\psi}=-\bigotimes_{l=1}^{N-1}U_{B_l}\ket{\psi}.
\ee
Now using Eq. \eqref{Bl} and putting in the state from \eqref{genstateN}, the above formula can be expressed as 
\be
  \sum_{i_0,i_1,\ldots, i_{N-1}=0,1}(-1)^{1+i_0+i_1+\ldots +i_{N-1}}\ket{i_0i_1\ldots i_{N-1}}\ket{\psi_{i_0i_1\ldots i_{N-1}}}=\sum_{i_0,i_1,\ldots, i_{N-1}=0,1}\ket{i_0i_1\ldots i_{N-1}}\ket{\psi_{i_0i_1\ldots i_{N-1}}}.
\ee
The above relation implies that the state \eqref{genstateN} is of the form
\be\label{state_sum=odd}
    \bigotimes_{l=1}^{N-1}U_{B_l}\ket{\psi}=  \sum_{i_0,i_1,\ldots, i_{N-1}=0,1}\ket{i_0i_1\ldots i_{N-1}}\ket{\psi_{i_0i_1\ldots i_{N-1}}} \quad \text{with} \quad i_0+i_1+\ldots +i_{N-1}=2n+1,
\ee
for some non-negative integer $n$.

Now, we consider relation \eqref{rellWN} for any $l$ and multiply it with $(\I-P_l)$ to obtain for any $l$
\begin{eqnarray}
    Z_A \otimes\I_l\otimes\left(\I-P_l\right) \ket{\psi}=\left(\I-P_l\right)\ket{\psi}.
\end{eqnarray}
As $B_{j,0}$ are certified in \eqref{Bl}, we obtain that $P_l=\proj{0\ldots0}\otimes\I_{B_1''\ldots B_{N-1}''}$. Now plugging in the state \eqref{state_sum=odd} in the above formula, we get that
\begin{eqnarray}
  \sum_{\substack{i_0,i_1,\ldots, i_{N-1}=0,1\\i_1+\ldots i_{N-1}-i_l\ne 0}} (-1)^{i_0}\ket{i_0i_1\ldots i_{N-1}}\ket{\psi_{i_0i_1\ldots i_{N-1}}}=\sum_{\substack{i_0,i_1,\ldots, i_{N-1}=0,1\\i_1+\ldots i_{N-1}-i_l\ne 0}} \ket{i_0i_1\ldots i_{N-1}}\ket{\psi_{i_0i_1\ldots i_{N-1}}}.
\end{eqnarray}
This implies from \eqref{state_sum=odd}  that 
\begin{eqnarray}\label{basisfix1}
    \ket{\psi_{1i_1\ldots i_{N-1}}}=0\quad \mathrm{s.t.} \quad i_1+\ldots+i_{N-1}=2n(n\ne0).
\end{eqnarray}
Now, considering \eqref{rellWNnew} for any $l$ and multiplying it with $\I-P_l$, we obtain
\begin{eqnarray}
    \I_A \otimes B_{l,0}\otimes\left(\I-P_l\right)\ket{\psi}=\left(\I-P_l\right)\ket{\psi}.
\end{eqnarray}
Now plugging in the state  \eqref{state_sum=odd} with the certified observables \eqref{Bl}, we obtain
\begin{eqnarray}
    \sum_{\substack{i_0,i_1,\ldots, i_{N-1}=0,1\\i_1+\ldots i_{N-1}-i_l\ne 0}} (-1)^{i_l}\ket{i_0i_1\ldots i_{N-1}}\ket{\psi_{i_0i_1\ldots i_{N-1}}}=\sum_{\substack{i_0,i_1,\ldots, i_{N-1}=0,1\\i_1+\ldots i_{N-1}-i_l\ne 0}}\ket{i_0i_1\ldots i_{N-1}} \ket{\psi_{i_0i_1\ldots i_{N-1}}}.
\end{eqnarray}
This implies from \eqref{state_sum=odd} that 
\begin{eqnarray}\label{basisfix2}
     \ket{\psi_{i_0i_1\ldots,1_l,\ldots, i_{N-1}}}=0\quad \mathrm{s.t.} \quad i_0+i_1+\ldots+i_{N-1}-i_l=2n(n\ne0).
\end{eqnarray}
It is now straightforward to observe from Eqs. \eqref{state_sum=odd},\eqref{basisfix1}, and \eqref{basisfix2} that the state $\ket{\psi}$ that attains the quantum bound of $\langle \mathcal{W}_N\rangle$ is of the form
\begin{eqnarray}\label{upstate}
     \bigotimes_{l=1}^{N-1}U_{B_l}\ket{\psi}=  \sum_{i_0,i_1,\ldots, i_{N-1}=0,1}\ket{i_0i_1\ldots i_{N-1}}\ket{\psi_{i_0i_1\ldots i_{N-1}}} \quad \text{with} \quad i_0+i_1+\ldots +i_{N-1}=1.
\end{eqnarray}

Let us now consider \eqref{rellWN} and multiply it with $P_l$ to obtain for any $l$
\begin{eqnarray}
    \left(\gamma_l X_A \otimes B_{l,1}
+\delta_lZ_A \otimes \I_{l}\right)\otimes P_l  \ket{\psi}=P_l\ket{\psi}.
\end{eqnarray}
Plugging in the state from \eqref{upstate} and observables from \eqref{Bl}, we obtain the following two conditions
\begin{align*}
    &\gamma_l\ket{\psi_{0\ldots 1_l\ldots 0_{N-1}}}-\delta_l\ket{\psi_{1\ldots 0_l\ldots 0_{N-1}}}=\ket{\psi_{1\ldots 0_l\ldots 0_{N-1}}}\\ 
    &\gamma_l\ket{\psi_{1\ldots 0_l\ldots 0_{N-1}}}+\delta_l\ket{\psi_{0\ldots 1_l\ldots 0_{N-1}}}=\ket{\psi_{0\ldots 1_l\ldots 0_{N-1}}}.
\end{align*}
Using the values of $\gamma_l,\delta_l$ from \eqref{coeffN}, the solution of these two relations is $\ket{\psi_{0\ldots 1_l\ldots 0_{N-1}}}=\frac{\alpha_{l+1}}{\alpha_{1}}\ket{\psi_{1\ldots 0_l\ldots 0_{N-1}}}$ for all $l=1,\ldots,N-1$.
Therefore, the state which maximally violated our steering inequalities has the following form  
\begin{align}\label{maxstateN}
    \left(\I_{AE}\otimes\bigotimes_{l=1}^{N-1}U_{B_l}\right)\ket{\psi}_{A\mathbf{B}E}=\left(\sum_{l=0}^{N-1}\alpha_{l+1}\ket{0\ldots,0,1_l,0,\ldots,0}\right)\otimes \left(\frac{1}{\alpha_1}\ket{\psi_{1,0,\ldots,0}}\right),
\end{align}
where $\ket{\xi}_{\mathbf{B}''E}=\frac{1}{\alpha_1}\ket{\psi_{1,0,\ldots,0}}$.
\end{proof}

\end{document}